\documentclass[12pt]{article}

\usepackage{graphicx}
\usepackage{amsmath}
\usepackage{amssymb}
\usepackage{amsthm}

\usepackage[colorlinks = true]{hyperref}

\hypersetup{
  pdftitle = {Doubled Color Codes},
  pdfauthor = {Sergey Bravyi, Andrew Cross}
}

\usepackage{xcolor}
\definecolor{darkred}  {rgb}{0.5,0,0}
\definecolor{darkblue} {rgb}{0,0,0.5}
\definecolor{darkgreen}{rgb}{0,0.5,0}
\hypersetup{
  urlcolor   = blue,         
  linkcolor  = darkblue,     
  citecolor  = darkgreen,    
  filecolor  = darkred       
}

\newcommand{\be}{\begin{equation}}
\newcommand{\ee}{\end{equation}}
\newcommand{\ba}{\begin{array}}
\newcommand{\ea}{\end{array}}
\newcommand{\bea}{\begin{eqnarray}}
\newcommand{\eea}{\end{eqnarray}}

\newcommand{\calA}{{\cal A }}
\newcommand{\calH}{{\cal H }}
\newcommand{\calD}{{\cal D }}

\newcommand{\calN}{{\cal N }}
\newcommand{\calB}{{\cal B }}
\newcommand{\calF}{{\cal F }}
\newcommand{\calG}{{\cal G }}
\newcommand{\calV}{{\cal V }}
\newcommand{\calE}{{\cal E }}
\newcommand{\calC}{{\cal C }}
\newcommand{\calS}{{\cal S }}

\newcommand{\calW}{{\cal W }}

\newcommand{\calT}{{\cal T }}
\newcommand{\calU}{{\cal U }}
\newcommand{\calO}{{\cal O }}

\newcommand{\FF}{\mathbb{F}}

\newcommand{\ZZ}{\mathbb{Z}}
\newcommand{\CC}{\mathbb{C}}

\newcommand{\trace}[1]{{\mathrm{Tr}{#1}}}
\newcommand{\css}[2]{{\mathrm{CSS}{({#1},{#2})}}}

\newcommand{\jbf}{{\mathbf j}}

\newcommand{\trn}[1]{{#1}^\intercal}
\newcommand{\supp}[1]{{\mathrm{supp}{(#1)}}}

\newtheorem{dfn}{Definition}

\newtheorem{lemma}{Lemma}
\newtheorem{fact}{Fact}

\newtheorem{corol}{Corollary}

\title{Doubled Color Codes}

\newcommand{\IBM}{IBM  T.J. Watson  Research Center, Yorktown Heights, NY 10598, USA}

\author{Sergey Bravyi\footnote{\IBM} \and Andrew Cross\footnotemark[1] }

\date{}

\begin{document}
\maketitle

\begin{abstract}
We show how to perform a fault-tolerant universal quantum computation in 2D architectures
using only transversal unitary operators and local syndrome measurements. 
Our approach is based on  a doubled version of the   2D color code.
It enables a transversal implementation of all logical gates
in the Clifford$+T$ basis using the gauge fixing method proposed recently by Paetznick and Reichardt.
The gauge fixing requires six-qubit parity measurements for
Pauli operators supported on faces of the honeycomb lattice with two qubits per site. 
Doubled color codes are promising candidates for the experimental
demonstration of logical gates since they do not require state distillation.
Secondly,  we propose a Maximum Likelihood algorithm for the error correction
and gauge fixing tasks  that enables a numerical  simulation 
of logical circuits in the Clifford$+T$ basis. The algorithm can be used in the online regime such that a new error syndrome
is revealed at each time step. 
We estimate the average number of logical gates that can be implemented reliably
for the smallest doubled color code and a toy noise model that includes
depolarizing memory errors and syndrome measurement errors.  
\end{abstract}

\newpage

\tableofcontents

\newpage

\section{Introduction}
\label{sec:intro}

Recent years have witnessed  several major steps  towards experimental
demonstration of quantum error correction~\cite{Barends2014,Kelly2015,corcoles2014}
giving us hope that a small-scale fault tolerant quantum memory may become a reality soon.
Quantum memories based on 
topological stabilizer codes such as the 2D surface code
are arguably among the most promising candidates
since they can tolerate a high level of noise 
and can be realized on a two-dimensional grid of qubits
with local parity checks~\cite{Dennis2001,Raussendorf2007,Fowler2009}.
Logical qubits encoded by such codes 
would be virtually isolated from the environment  
by means of an active error correction
and could  preserve delicate superpositions of  quantum states for
extended periods of time. 

Meanwhile, demonstration of a  universal set of
logical gates required for a fault-tolerant quantum computing
remains a distant goal. Although the surface code provides
a low-overhead implementation of  logical  Clifford gates such as 
the CNOT or the Hadamard gate~\cite{Raussendorf2007,Fowler2009},
implementation of logical non-Clifford gates 
poses a serious challenge.
  Non-Clifford gates such as the single-qubit $45^\circ$ phase shift known as the $T$-gate
  are required to express interesting quantum algorithms
but  their operational cost in the surface code architecture 
exceeds the one of Clifford gates 
by orders of magnitude. This large overhead stems from the 
state distillation subroutines
which may require a  thousand or more physical qubits to realize just a single
logical $T$-gate~\cite{Jones2013,Fowler2013surface}.
Some form of state distillation is used by all currently known fault-tolerant protocols
based on 2D stabilizer codes. 

The purpose of this paper is to propose an alternative family of quantum codes and fault tolerant protocols
for 2D architectures where all logical gates are implemented 
transversally. Recall that a logical gate is called transversal if it 
can be implemented by applying some single-qubit
rotations to each physical qubit. 
Transversal gates are highly desirable since they introduce no overhead
and do not spread errors. 
Assuming that all  qubits are controlled in parallel, a transversal 
gate takes the same time as a single-qubit rotation,
which is arguably the best one can hope for. 
Unfortunately, transversal gates have a very limited computational power. 
A no-go theorem proved by  Eastin and Knill~\cite{Eastin2009} asserts
that a quantum code can have only a finite number of transversal gates
which rules out universality. In the special case of 2D stabilizer codes
a more restrictive version of this  theorem have been proved
asserting that transversal logical gates must belong to the Clifford 
group~\cite{Bravyi2013,Pastawski2015,Beverland2014}.

To circumvent these no-go theorems we employ the gauge fixing method
proposed  recently by Paetznick and Reichardt~\cite{Paetznick2013}.
A  fault-tolerant protocol based on the gauge fixing method alternates
between two error correcting codes that provide a transversal implementation
of logical Clifford gates and the logical $T$-gate respectively. Thus a computational universality 
is  achieved by combining transversal gates of two different codes. 
A  conversion between the codes can be made fault-tolerantly if their stabilizer groups have a sufficiently large intersection. 
This is achieved by properly choosing a  pattern of parity checks 
measured at each time step and applying  a gauge fixing operator depending on the measured syndromes.
The latter is responsible both for error correction and for switching between two different
encodings of the logical qubit.

Our goal for the first part of the paper (sections~\ref{sec:subs}-\ref{sec:CT15}) is 
to develop effective error models  and decoding algorithms 
suitable for  simulation of  logical circuits in  the Clifford$+T$ basis.
Although a transversal implementation of logical $T$-gates offers a substantial
overhead reduction, it poses several challenges for the  decoding algorithm.
First, $T$-gates introduce correlations between $X$-type and $Z$-type errors
that cannot be described by the standard stabilizer formalism.
This may prevent the decoder from using  error syndromes measured 
before application of a $T$-gate in the error correction steps performed afterwards.
 Second,  implementation of $T$-gates by the gauge fixing method
requires an online decoder  such that a new gauge fixing operator has to be 
computed and applied prior to each logical $T$-gate. 
Thus a practical decoder must have  running time $O(1)$ per logical gate
independent of the total length of the circuit. 
The present work makes two contributions
that partially address these challenges. 

First, we generalize the stabilizer formalism  commonly used for
a numerical  simulation of  error correction to logical Clifford$+T$ circuits.
 Specifically, we show how to commute Pauli errors through  a composition of a  transversal $T$-gate
and a certain  twirling map such that 
the effective error model at each step of the circuit can be described by  random Pauli errors,
even though the circuit may contain many non-Clifford gates. 
The twirling map has no effect on the logical state since it includes only stabilizer operators.
We expect that this technique may find applications in  other contexts.

Second, we propose a Maximum Likelihood (ML) decoding  algorithm
for the error correction and gauge fixing tasks. 
The ML decoder applies the Bayes rule to find a recovery operator which is most likely
to succeed in a given task based on the full  history of  measured syndromes. This is achieved by 
properly  taking into account  statistics of memory and measurement errors, as well as
correlations between $X$-type and $Z$-type errors introduced by transversal $T$-gates.
Although the number of syndromes that the  ML decoder has to process scales
linearly with the length of the logical circuit, the decoder has 
a constant  running time per logical gate which scales as $O(n2^n)$ for
a code with $n$ physical qubits. The decoder can be used in the online
regime for sufficiently small codes, which is
crucial for the future experimental demonstration of logical gates. 
A key ingredient of the decoder is the fast Walsh-Hadamard transform. 
A heuristic  approximate version of the algorithm
called a sparse ML decoder is proposed 
that could be applicable to medium size codes.

We apply the ML decoder to a particular gauge fixing protocol 
proposed by Anderson et al~\cite{Anderson2014}.
The protocol alternates between the $15$-qubit Reed-Muller code
and the $7$-qubit Steane code that provide a transversal implementation of the
$T$-gate and Clifford gates respectively.
Numerical simulations are performed 
for a phenomenological error model that consists of depolarizing memory errors and 
syndrome measurement errors with some rate $p$. 
Following ideas of a randomized benchmarking~\cite{Knill2008,Chow2009} we choose a Clifford$+T$ circuit
at random, such that
each Clifford gate is drawn from the uniform distribution on the single-qubit Clifford group. 
The circuit alternates between Clifford and $T$ gates. 
The quantity we are interested in is 
a logical error rate defined as  $p_L=1/g$,
where $g$ is the average number of logical gates implemented before 
the first failure in the error correction or gauge fixing subroutines.
Here $g$ includes both Clifford and $T$ gates. 
The sparse ML decoder enables a numerical simulation of circuits with more than 10,000 logical gates. 
For small error rates we observed  a scaling $p_L =Cp^2$ with  $C\approx 182$.
Assuming that a physical Clifford$+T$ circuit
has an error probability $p$ per gate,   the logical circuit becomes more reliable than the 
physical one  provided  that  $p_L<p$, that is, $p<p_0=C^{-1} \approx 0.55\%$.
This value can be viewed as an ``error threshold" of the proposed protocol. 
The observed threshold is comparable with the one calculated by 
Brown et al~\cite{Brown2015} for a gauge fixing protocol based on the 3D color
codes which were recently proposed by Bombin~\cite{Bombin2015}.
We note however that Ref.~\cite{Brown2015} studied only 
the storage of a logical qubit (no logical gates). 
We anticipate that the tools developed in the present paper could be
used to simulate logical Clifford$+T$ circuits based on the 3D color codes as well.
It should be pointed out that the threshold $p_0=0.55\%$ is almost one order of magnitude smaller
than the one of the 2D surface code for the analogous error model~\cite{Wang2003}.
This is the price one has to pay for the low-overhead implementation of all logical gates. 

The numerical results obtained for the $15$-qubit code call for more general code constructions
that could  achieve a more favorable scaling of the logical error rate.
In the second part of the paper  (sections~\ref{sec:summary}-\ref{sec:gadget})
we propose an infinite family of 2D quantum codes with  a diverging code distance 
that enable  implementation of  Clifford$+T$ circuits by the gauge fixing method.
The number of physical qubits required to achieve a code distance $d=2t+1$
is $n=2t^3+8t^2+6t-1$.
For comparison, the 3D color codes of Ref.~\cite{Bombin2015} require
$n=4t^3+6t^2+4t+1$ physical qubits. 
The new codes can be embedded into the 2D honeycomb lattice with two qubits per site
such that all syndrome measurements required for error correction and gauge fixing are
spatially local. More precisely, any check operator measured in the protocol
acts on at most six qubits located on some  face of the lattice. 
As was pointed out in Ref.~\cite{Bombin2015}, the gauge fixing method circumvents
no-go theorems proved for transversal  non-Clifford gates in the 
2D geometry~\cite{Bravyi2013,Pastawski2015,Beverland2014}
since the decoder that controls all quantum operation may perform a non-local
classical processing.

The key ingredient of our approach is a doubling transformation 
from the classical coding theory
originally proposed 
by Betsumiya and Munemasa~\cite{Betsum2010}.
Its quantum analogue can be used to construct  high-distance codes
with a special symmetry  required for  transversality  
of logical $T$-gates. Namely, a quantum code of CSS type~\cite{CSS1996,Steane1996}
 is said to be triply even (doubly even)
if the weight of any $X$-type stabilizer  is a multiple of eight (multiple of four).
Any triply even CSS code of odd length is known to have a transversal $T$-gate.
The doubling transformation combines a triply even code with distance $d-2$ and
two copies of  a doubly even  code with distance $d$ to produce  a triply even code
with distance $d$. Our construction recursively applies the doubling transformation
to the family of regular color codes on the honeycomb lattice~\cite{Bombin2006}
such that each recursion level increases the code distance by two. 
The regular color codes are known to be doubly even~\cite{Kubica2015} (in a certain generalized sense).
Producing a distance-$d$ triply even code  
requires $(d-1)/2$ recursion levels
that combine  color codes with distance $3,5,\ldots,d$.
We refer to the new family of codes as doubled color code since
the construction relies on taking two copies of the regular color codes. 
It should not be confused with the quantum double construction
from the topological quantum field theory~\cite{Kitaev2003}.
Doubled color codes have almost all properties required for
implementation of logical Clifford$+T$ circuits by the gauge fixing method.
Namely, a doubled color code has a transversal $T$-gate and can be
converted fault-tolerantly to the regular color code which is known to have  transversal Clifford gates~\cite{Kubica2015}.
Unfortunately, the doubling transformation does not preserve spatial locality
of check operators. Even worse, some check operators of a doubled color code
have very large weight and their syndromes cannot be measured in a fault-tolerant fashion.
Converting the doubled color codes into
a local form is our main technical contribution.
This is achieved in two steps. First we 
show how to implement all levels of the recursive doubling transformation
on the honeycomb lattice with two qubits per site such that almost
all check operators of the output code are spatially local. 
We show that each of the remaining non-local checks can be decomposed into a product
of local ones by introducing several ancillary qubits and extending the code
properly to the ancillary qubits. 
This technique is reminiscent of perturbation theory gadgets that are used to
generate effective low-energy Hamiltonians with long-range many-body interactions 
starting  from a simpler high-energy Hamiltonian with 
short-range  two-body interactions~\cite{Oliveira2008}.

The $15$-qubit code studied in Section~\ref{sec:CT15} can be viewed as  the smallest
example of a doubled color code. Furthermore, the $49$-qubit triply even code 
with distance $d=5$
discovered by an exhaustive numerical search in Ref.~\cite{Bravyi2012}
can be viewed as a doubled color code 
obtained from the regular color code $[[17,1,5]]$ on the 
square-octagon lattice via the doubling transformation.
The $49$-qubit code is  optimal in the sense that no distance-$5$
code with less than $49$ qubits can be triply even~\cite{Betsum2010,Bravyi2012}.
Thus the family of doubled color codes includes
the best known examples of triply even codes with a small distance.

Although this paper focuses on codes with a single logical qubit,
our fault-tolerant  protocols  can be incorporated into the
lattice surgery method based on the regular color codes~\cite{Landahl2014}.
The former would provide implementation of logical single-qubit rotations
decomposed into a product of Clifford and $T$-gates while the latter 
enables logical CNOTs  and can serve as a quantum memory.
We note that efficient and nearly optimal algorithms for decomposing single-qubit rotations into
a product of Clifford and $T$-gates  have been proposed recently~\cite{Kliuchnikov2013fast,Selinger2015}.

To make the paper self-contained, 
we provide all necessary background on quantum codes of CSS type, the 
gauge fixing method, and transversal logical gates in Sections~\ref{sec:notations}-\ref{sec:trans}.
The effective error model describing the action of transversal $T$-gates on random Pauli
errors is developed in Section~\ref{sec:Tgate}.
We describe the ML decoding algorithm suitable for simulation of Clifford$+T$ circuits
in Section~\ref{sec:ML}.
Numerical simulation of random Clifford$+T$ logical circuits based on the family
of $15$-qubit codes is  described in Section~\ref{sec:CT15}.
This section also serves as an example illustrating  the general construction of doubled color codes.
The latter is described in Sections~\ref{sec:summary}-\ref{sec:gadget}.
First, we highlight main properties of the doubled color codes  in Section~\ref{sec:summary}.
Definition of regular 2D color codes and their properties are summarized in Section~\ref{sec:color}.
Doubled color codes and their embedding into the honeycomb lattice
are defined in Sections~\ref{sec:doubling},\ref{sec:color2}.
The most technical part of the paper is Section~\ref{sec:gadget} explaining
how to convert doubled color codes into a spatially local form.

\section{Notations}
\label{sec:notations}

Let $\FF_2^n$ be the $n$-dimensional linear space over the binary field $\FF_2$. 
A vector $x\in \FF_2^n$ is regarded as a column vector with 
components $x_1,\ldots,x_n$. We shall write $\trn{x}$ for the corresponding
row vector. The set of integers $\{1,2,\ldots,n\}$ will be denoted $[n]$.
Let $\supp{x}\subseteq [n]$ be the
support of $x$, that is, the
subset of indexes $i$ such that $x_i=1$.
We shall often identify a vector $x$ and 
the subset $\supp{x}$.
Let $|x|\equiv |\supp{x}|$ be the weight of $x$, that is, the number of non-zero components.
Given a subset $A\subseteq [n]$ 
let $x_A\in \FF_2^{|A|}$ be a restriction of $x$ onto $A$, 
that is, a vector  obtained from $x$ by deleting all components $x_i$
with $i\notin A$. Conversely, given a vector $y\in \FF_2^{|A|}$
let $y[A]\in \FF_2^n$ be a vector obtained from $y$ by inserting zero components
for all $i\notin A$. For example, if $n=5$, $A=\{1,3,5\}$, and $y=(111)$ then
$y[A]=(10101)$.
The space $\FF_2^n$ is equipped with a standard basis $e^1,e^2,\ldots,e^n$, where
$e^j\equiv 1[j]$ is the vector with a single non-zero component indexed by $j$.
We shall use notations $\overline{0}$ and $\overline{1}$ for the
all-zeros and the all-ones vectors.
The inner product between vectors $x,y\in \FF_2^n$ is defined as
\[
x^\intercal y=\sum_{i=1}^n x_i y_i {\pmod 2}.
\]
A linear subspace spanned by vectors $x^1,\ldots,x^m\in \FF_2^n$ will
be denoted $\langle x^1,\ldots,x^m\rangle$.
Given a subset $A\subseteq [n]$ and a subspace $\calS\subseteq \FF_2^{|A|}$, let
$\calS[A]\subseteq \FF_2^n$ be the subspace spanned by vectors $y[A]$ with
$y\in \calS$.
Let $\calE^n$ and $\calO^n$ be the subspaces of $\FF_2^n$ spanned by all even-weight 
and all odd-weight vectors respectively.
We shall use shorthand notations $\calE\equiv \calE^n$ and $\calO\equiv \calO^n$
whenever the value of $n$ is clear from the context. 
Given a linear subspace $\calS\subseteq \FF_2^n$ let $\calS^\perp$ be the orthogonal
subspace, 
\[
\calS^\perp = \{ x\in \FF_2^n \, : \, \trn{x} y=0 \quad \mbox{for all $y\in \calS$}\}
\]
and $\dot{\calS}$ be the subspace spanned by all even-weight vectors orthogonal to $\calS$,
\[
\dot{\calS}=\calS^\perp \cap \calE.
\]
A subspace $\calS$ is self-orthogonal if $\calS\subseteq \calS^\perp$.
We shall use identities 
$(\calS^\perp)^\perp=\calS$,   $(\calS+\calT)^\perp =\calS^\perp\cap \calT^\perp$,
and  $\calE^\perp=\langle \overline{1}\rangle$.
Given a subspace $\calS$, let $d(\calS)$ be the minimum weight of odd-weight
vectors in $\calS^\perp$, 
\begin{equation}
\label{distance}
d(\calS)\equiv \min{ \{ |f| \, : \, f\in \calS^\perp \cap \calO\}}.
\end{equation}

Consider now a system of $n$ qubits and let $X_j,Y_j,Z_j$ be the 
Pauli operators  acting on a qubit $j$ tensored with the identity
on the remaining qubits. 
Given a vector $f\in \FF_2^n$ and 
a single-qubit Pauli operator $P$ 
let $P(f)$ be the $n$-qubit operator that applies $P$
to each qubit in the support of $f$, 
\[
P(f)=\prod_{j\in \supp{f}} P_j.
\]
For any subset $\calS\subseteq \FF_2^n$ define
the corresponding set of Pauli operators
  $P(\calS)\equiv \{ P(f) \, : \, f\in \calS\}$.
A pair of subspaces $\calA,\calB\subseteq \FF_2^n$ defines a group of $n$-qubit Pauli operators
\[
\css{\calA}{\calB}= \langle X(\calA),Z(\calB)\rangle.
\]
Any element of $\css{\calA}{\calB}$ has a form $i^m X(f) Z(g)$ for some
$f\in \calA$,  $g\in \calB$, and some integer $m$. 
Note that the  group $\css{\calA}{\calB}$ is abelian iff $\calA$ and $\calB$ are mutually orthogonal,
 $\calA\subseteq \calB^\perp$, since  
\[
X(f)Z(g)=(-1)^{\trn{f} g} \,Z(g)X(f).
\]
Given a vector $f\in \FF_2^n$, let 
$|f\rangle=|f_1\otimes \cdots \otimes f_n\rangle$ be the corresponding  basis state of $n$ qubits.
Note that $|f\rangle=X(f)|\overline{0}\rangle$.

\section{Subsystem quantum codes and gauge fixing}
\label{sec:subs}

This section summarizes some known facts concerning subsystem
quantum codes of Calderbank-Steane-Shor (CSS) type~\cite{CSS1996,Steane1996} and the 
gauge fixing method. 
A quantum CSS code   is constructed from a pair of linear subspaces $\calA,\calB\subseteq \FF_2^n$
that are mutually orthogonal,  $\calA\subseteq \calB^\perp$.
Such pair defines an abelian group of Pauli operators
$\calS=\css{\calA}{\calB}$ called a {\em stabilizer group}.
Elements of $\calS$ are called stabilizers.  
We shall often identify a CSS code and its stabilizer group.
A subspace spanned by $n$-qubit states 
invariant under the action of any stabilizer is called a {\em codespace}.
A projector onto the codespace can be written as
\begin{equation}
\label{proj}
\Pi=\frac1{|\calS|}\sum_{G\in \calS} G.
\end{equation}
In this paper we  only consider a restricted class  of CSS codes
such that all vectors in $\calA$ and $\calB$ have even weight
whereas the number of physical qubits $n$ is odd,
\begin{equation}
\label{simpleCSS}
n=1{\pmod 2}, \quad \quad \calA\subseteq \calE, \quad \quad \calB\subseteq \calE.
\end{equation}
We shall only consider codes with a single logical (encoded) qubit.
Operators acting on the encoded qubit are expressed in terms of 
{\em logical Pauli operators} 
\begin{equation}
\label{XYZL}
X_L=X(\overline{1}), \quad Y_L=Y(\overline{1}), \quad Z_L=Z(\overline{1}).
\end{equation}
Logical operators commute with any stabilizer due to Eq.~(\ref{simpleCSS}) and thus preserve the codespace.
Furthermore, $X_L,Y_L,Z_L$ obey the same commutation rules as single-qubit Pauli 
operators $X,Y,Z$ respectively. 
More generally, a Pauli operator $P$ is called a logical operator
iff it coincides with $X_L$, $Y_L$, or $Z_L$ modulo stabilizers, that is,
$P=X(a)Z(b)$, where $a\in \calA+\alpha \overline{1}$,
$b\in \calB+\beta\overline{1}$, and at least one of 
the coefficients $\alpha,\beta$ is non-zero. 
A logical state  encoding a single-qubit state
$\eta=(1/2)(I+\alpha X+\beta Y+\gamma Z)$ is  defined as  
\begin{equation}
\label{rhoL}
\rho_L\equiv \rho_L(\eta)=O_L \Pi, \quad \quad O_L=\delta(I+ \alpha X_L + \beta Y_L + \gamma Z_L),
\end{equation}
where $\delta$ is a  coefficient responsible for normalization $\trace{(\rho_L)}=1$.

Pauli operators  that commute with both stabilizers and logical operators generate a group
$\calG$ called a {\em gauge group}.  Elements of $\calG$
are called gauge operators. Note that a Pauli operator $X(f)$ 
commutes with all stabilizers iff $f\in \calB^\perp$. Likewise, $X(f)$ commutes
with the logical operators iff $\trn{\overline{1}}f=0$, that is, $f\in \calE$. 
Thus $X(f)$ is a gauge operator iff $f\in \dot{\calB}$, where we use notations
of Section~\ref{sec:notations}. The same reasoning shows that $Z(g)$ is a gauge 
operator iff $g\in \dot{\calA}$. Thus the gauge group 
corresponding to a stabilizer group $\calS=\css{\calA}{\calB}$ 
is given by
\begin{equation}
\label{Ggroup}
\calG=\css{\dot{\calB}}{\dot{\calA}}.
\end{equation}
By definition, $\calS\subseteq \calG$.
Note that $E\rho_LE^\dag=\rho_L$ for any $E\in \calG$, that is, 
gauge operators have no effect on the logical state.

A noise maps the logical state $\rho_L$ to a  probabilistic mixture of
states $E_\alpha \rho_L E_\alpha^\dag$, where $E_\alpha$ are some $n$-qubit Pauli operators
called {\em memory errors} or simply errors. We shall only consider noise that can be described
by random Pauli errors. 
Note that errors $E_\alpha$ and $E_\beta$ have the same
action on any logical state whenever $E_\beta^\dag E_\alpha\in \calG$, 
that is, gauge-equivalent  errors can be identified.  
We shall say that an error $E$ is {\em non-trivial} if $E\notin \calG$.
Errors are diagnosed by measuring eigenvalues of some stabilizers. 
An eigenvalue measurement of a  stabilizer $X(f)$  
has an outcome $(-1)^{\xi(f)}$, where $\xi(f)\in \FF_2$ is called a {\em syndrome} of $X(f)$.
A corrupted state  $E\rho_L E^\dag$
with a memory error $E=X(a)Z(b)$ has a syndrome
$\xi(f)=\trn{f}b$. It reveals whether $E$ commutes or anti-commutes with $X(f)$. 
A syndrome of a 
stabilizer $Z(g)$   is defined as 
$\zeta(g)=\trn{g}a$. It reveals whether $E$ commutes or anti-commutes with  $Z(g)$. 
By definition, gauge-equivalent errors have the same syndromes.

In some cases stabilizer syndromes  cannot be measured directly (for example, if a
stabilizer has too  large weight) but they can be inferred by measuring
eigenvalues of some gauge operators. For example, suppose a stabilizer $X(f)$
can be represented as a product of $X$-type gauge operators $X(g^\alpha)$, that is,
$f=g^1+g^2+\ldots+g^m$ for some $f\in \calA$ and $g^\alpha\in \dot{\calB}$.
An eigenvalue measurement of the gauge operator $X(g^\alpha)$ 
has an outcome $(-1)^{\xi(g^\alpha)}$, where $\xi(g^\alpha)$ is called 
a {\em gauge syndrome}. Once all gauge syndromes $\xi(g^\alpha)$ have been
measured, the syndrome of $X(f)$ is inferred from $\xi(f)=\xi(g^1)+\ldots+\xi(g^m)$.
Since gauge operators commute with stabilizers, these measurements do not
affect syndromes of any $Z$-type stabilizers. However, measuring gauge syndromes
may change the encoded state and the latter no longer has a form
$E_\alpha \rho_L E_\alpha^\dag$ (since some gauge degrees of freedom have been fixed).
We shall assume that each syndrome measurement is followed by a twirling map
\[
\calW_\calG(\rho)=\frac1{|\calG|} \sum_{G\in \calG} G\rho G^\dag
\]
that applies a random gauge operator $G$ drawn from the uniform distribution on $\calG$.
The twirling map restores the original form of the encoded state $E_\alpha \rho_L E_\alpha^\dag$
by bringing all gauge degrees of freedom to the maximally mixed state. 
Gauge syndromes $\zeta(g^\alpha)$ for $Z$-type gauge operators $g^\alpha\in \dot{\calA}$
are defined analogously. 

A memory error $E$ is said to be {\em undetectable} if it 
commutes with any element of $\calS$. Equivalently, $E$
has zero syndrome for any stabilizer.
Errors $E$ that are both undetectable and non-trivial ($E\notin \calG$)
should be avoided since they  can alter the logical state.
A {\em code distance} $d$ is defined as the minimum weight of an undetectable
non-trivial error. The code $\css{\calA}{\calB}$  has distance
\begin{equation}
\label{dAdB}
d=\min{\{ d(\calA),d(\calB) \}},
\end{equation}
where $d(\calA)$ and $d(\calB)$ are defined by Eq.~(\ref{distance}).

Next let us discuss a {\em gauge fixing} operation that 
extends the stabilizer group  to a larger group or, equivalently,
reduces the gauge group to a smaller subgroup. 
Consider  a pair of codes with stabilizer groups $\calS=\css{\calA}{\calB}$
and $\calS'=\css{\calA'}{\calB'}$ such that $\calA\subseteq \calA'$ and $\calB\subseteq \calB'$.
Let $\calG=\css{\dot{\calB}}{\dot{\calA}}$ and $\calG'=\css{\dot{\calB}'}{\dot{\calA}'}$
be the respective gauge groups. 
We note
that $\calS\subseteq \calS'$ and $\calG\supseteq \calG'$.
Let us represent $\calG$ as a disjoint union of cosets
$G_a\calG'$ for some fixed set of coset representatives $G_1,\ldots,G_m\in \calG$. 
Simple algebra shows that the codespace projectors $\Pi$ and $\Pi'$ of the two codes
are related as 
\begin{equation}
\label{rhoL1}
\Pi= \sum_{a=1}^m G_a \Pi' G_a^\dag,
\end{equation}
and all subspaces $G_a \Pi' G_a^\dag$ are pairwise orthogonal.
Substituting Eq.~(\ref{rhoL1}) into Eq.~(\ref{rhoL}) 
shows that  any logical state $\rho_L$ of the code $\calS$ can be written as
\begin{equation}
\label{rhoL2}
\rho_L =\frac1m \sum_{a=1}^m G_a \rho_L' G_a^\dag,
\end{equation}
where  $\rho_L'$ is a logical state of the code $\calS'$.
Furthermore, $\rho_L$ and $\rho_L'$ encode the same state.
Since $G_a$ are not gauge operators for the code $\calS'$, they should be
treated as memory errors. Moreover, any pair of errors $G_a$
and $G_b$ can be distinguished by measuring a complete syndrome
of the code $\calS'$ (that is, measuring syndromes of 
some complete set of generators of $\calS'$).
Indeed, if $G_a$ and $G_b$ would have the same syndromes
for some $a\ne b$ 
then the operator $G_aG_b^\dag$ would commute with any
element of $\calS'$ as well as  with the logical operators.
This would imply $G_aG_b^\dag\in \calG'$ which is a contradiction since
$G_a\calG'\ne G_b\calG'$.
Thus a complete syndrome measurement for the code $\calS'$ projects the 
state $\rho_L$ in Eq.~(\ref{rhoL2}) onto some particular term $G_a \rho_L' G_a^\dag$
and the coset $G_a\calG'$ can be inferred from the measured syndrome.
Applying a recovery operator $R$  chosen as any representative of the
coset $G_a\calG'$ yields a state
$RG_a\rho_L' G_a^\dag R^\dag = \rho_L'$. 
This completes the gauge fixing operation.
The gauge fixing can be performed in the reverse direction as well.
Namely, the encoded state $\rho_L'$ can be mapped to $\rho_L$
by applying randomly one of the operator $G_1,\ldots,G_m$ drawn from the uniform
distribution, see Eq.~(\ref{rhoL2}). 
We shall collectively refer to the gauge fixing and its reverse as {\em code deformations}.

It will be convenient to distinguish between {\em subsystem} and  {\em regular} CSS codes
(a.k.a. subspace or stabilizer codes). 
The latter is a special class of CSS codes such that the stabilizer and the gauge groups
are the same. In other words, a stabilizer group  $\css{\calA}{\calB}$ defines
a regular CSS code  iff
$\calA=\dot{\calB}$ and $\calB=\dot{\calA}$.
Note that applying the dot operation twice gives the original subspace,
$\ddot{\calA}=\calA$,  whenever $\calA\subseteq \calE$ and $n$ is odd.
Thus $\calA=\dot{\calB}$ iff $\calB=\dot{\calA}$.
A regular CSS code has a  two-dimensional codespace with  an orthonormal basis
\begin{equation}
\label{logical01}
|0_L\rangle =|\calA|^{-1/2} \sum_{f\in \calA} |f\rangle
\quad \mbox{and} \quad
|1_L\rangle=X_L|0_L\rangle =|\calA|^{-1/2}  \sum_{f\in \calA} |f+\overline{1}\rangle,
\end{equation}
such that $\Pi=|0_L\rangle\langle 0_L|+|1_L\rangle\langle 1_L|$.
In the case of subsystem CSS codes, the gauge group is strictly larger than the stabilizer group,
so that the codespace can be further decomposed into a tensor product of a logical
subsystem  and a gauge subsystem.

\section{Transversal logical gates}
\label{sec:trans}

In this section we state sufficient conditions under which 
a quantum code of CSS type has   transversal  logical gates.
Recall that a single-qubit unitary gate $V$ is called transversal if 
there exist a product unitary operator
 $U_{all}=U_1\otimes U_2\otimes \cdots \otimes U_n$
that  preserves the codespace and the action of $U_{all}$ on the
logical qubit coincides with $V$. 
More precisely,  let $\rho_L(\eta)$ be a logical state 
that encodes a single-qubit state $\eta$,
see Eq.~(\ref{rhoL}).
Then we require that $U_{all}\, \rho_L(\eta) U_{all}^{-1} = \rho_L(V\eta V^{-1})$
for all $\eta$. 
To simplify notations, we shall write $U_{all}=\prod_{j=1}^n U_j$ meaning that $U_j$
acts on the $j$-th qubit. 
We shall use the following set of single-qubit logical gates:
\[
H=\frac1{\sqrt{2}} \left[ \ba{cc} 1 & 1 \\ 1 & -1 \\ \ea\right],
\quad 
S=\left[ \ba{cc} 1 & 0 \\ 0 & i \\ \ea\right], \quad \mbox{and} \quad 
T=\left[ \ba{cc} 1 & 0 \\ 0 & e^{i\pi/4} \\ \ea\right].
\]
It is known that $H$ and $S$ generate the full Clifford group on one qubit,
whereas $H,S,T$ is a universal set generating a dense subgroup of the unitary group.
The set $H,S,T$ is known as the Clifford$+T$ basis. 
To state sufficient conditions for transversality we shall need  notions
of  a doubly-even and a  triply-even subspace~\cite{Betsum2010}.
\begin{dfn}
A  subspace $\calA\subseteq \FF_2^n$ is 
doubly-even iff there exist disjoint subsets $M^\pm \subseteq [n]$ such that 
\begin{equation}
\label{even2}
 |f\cap M^+| - |f\cap M^-| = 0{\pmod 4} \quad\quad  \mbox{for all $f\in \calA$}.
\end{equation}
\end{dfn}
\begin{dfn}
A  subspace $\calA\subseteq \FF_2^n$ is 
triply-even iff there exist disjoint subsets $M^\pm \subseteq [n]$ such that 
\begin{equation}
\label{even3}
 |f\cap M^+| - |f\cap M^-| = 0{\pmod 8} \quad\quad  \mbox{for all $f\in \calA$}.
\end{equation}
\end{dfn}

We require that at least one of the subsets $M^\pm$ 
is non-empty, since otherwise the definitions are meaningless.
Below we implicitly assume that each doubly- or  triply-even 
subspace is equipped with a pair of subsets $M^\pm$ satisfying Eq.~(\ref{even2})
or Eq.~(\ref{even3}) respectively.  In this case
$\calA$ is said to be doubly- or triply-even  with respect to  $M^\pm$.
The original definition of triply-even codes given in~\cite{Betsum2010}
is recovered by choosing $M^+=[n]$ and $M^-=\emptyset$. 

Consider a pair of disjoint subsets $M^\pm \subseteq [n]$
such that 
\[
m=|M^+|-|M^-|=1{\pmod 2}
\]
and define $n$-qubit product operators
\begin{equation}
\label{TSH}
T_{all}=\prod_{j \in M^+} T_j \cdot \prod_{j\in M^-} T_{j}^{-1},
\quad
S_{all}=\prod_{j \in M^+} S_j \cdot \prod_{j\in M^-} S_{j}^{-1},
\quad
\mbox{and}  \quad H_{all}=\prod_{j=1}^n H_j.
\end{equation}
The following lemma provides sufficient conditions for transversality 
of $T,S$, and $H$ gates. 
\begin{lemma}
\label{lemma:transversal}
Consider a quantum code $\css{\calA}{\calB}$, where $\calA,\calB\subseteq \FF_2^n$
are mutually orthogonal subspaces satisfying Eq.~(\ref{simpleCSS}). 
The code has a transversal  gate $T^m$, 
\begin{equation}
\label{Tall}
T_{all}\, \rho_L(\eta)T_{all}^{-1} = \rho_L(T^m \eta T^{-m}),
\end{equation}
whenever $\calB=\dot{\calA}$ and  $\calA$ is triply even with respect to 
$M^\pm$.
The code has a  transversal  gate $S^m$,
\begin{equation}
\label{Sall}
S_{all} \, \rho_L(\eta)S_{all}^{-1} = \rho_L(S^m \eta S^{-m}),
\end{equation}
whenever $\calA\subseteq \calB$ and  $\calA$ is doubly even with respect to
 $M^\pm$. 
The code has a transversal $H$-gate,
\begin{equation}
\label{Hall}
H_{all}\, \rho_L(\eta)H_{all} = \rho_L(H \eta H),
\end{equation}
 whenever $\calA=\calB$.
\end{lemma}

\begin{proof}
We start from Eq.~(\ref{Tall}). Since $\calB=\dot{\calA}$, the
codespace is two-dimensional with an orthonormal basis
$|0_L\rangle$, $|1_L\rangle$ defined in Eq.~(\ref{logical01}).
Clearly, $T_{all} |f\rangle = e^{i\pi/4 (|f\cap M^+| - |f\cap M^-|)} |f\rangle$
for any basis state $|f\rangle$ of $n$-qubits. 
Combining this identity and definition of the logical state $|0_L\rangle$
one gets
\[
T_{all}|0_L\rangle=T_{all}\sum_{f\in \calA} |f\rangle =
\sum_{f\in \calA} e^{i\pi/4 (|f\cap M^+| - |f\cap M^-|)} |f\rangle
=|0_L\rangle.
\]
Here we ignored the normalization of $|0_L\rangle$.
Using the identity $|(f+\overline{1})\cap M^\pm|=|M^\pm|-|f\cap M^\pm|$
and definition of the logical state
$|1_L\rangle$ one gets 
\[
T_{all}|1_L\rangle=T_{all}\sum_{f\in \calA} |f+\overline{1}\rangle =
e^{i\pi m/4}\sum_{f\in \calA} e^{-i\pi/4 (|f\cap M^+| - |f\cap M^-|)} |f+\overline{1}\rangle
=e^{i\pi m/4}|1_L\rangle.
\]
This proves Eq.~(\ref{Tall}). 
Consider now Eq.~(\ref{Sall}). Let $\calH \equiv \css{\calA}{\calB}$
be the stabilizer group.
 First we claim that $\calH$
  is invariant under the conjugated action 
of $S_{all}$, that is, 
\begin{equation}
\label{VSVdag=S}
S_{all} \, \calH S_{all}^{-1}=\calH.
\end{equation}
Indeed,  since $S_{all}$ commutes with $Z$-type stabilizers,
 it suffices to check that $S_{all} X(f)S_{all}^{-1} \in \calH$ for all $f\in \calA$.
Using the identities $SXS^{-1} =iXZ$ 
and $S^{-1} X S=-iXZ$  one gets
\[
S_{all} X(f)S_{all}^{-1}  =i^{|f\cap M^+| - |f\cap M^-|}  X(f) Z(f)=X(f) Z(f)\in \calH.
\]
Here we used the assumption that $\calA$ is doubly even and 
$\calA\subseteq \calB$ which implies $Z(f)\in \calH$ for any $f\in \calA$.
This proves Eq.~(\ref{VSVdag=S}). We conclude that $S_{all}$ preserves the codespace,
$S_{all}\, \Pi S_{all}^{-1} =\Pi$. Therefore
\[
S_{all} \, X_L \Pi S_{all}^{-1} = S_{all} \, X_L S_{all}^{-1} \Pi = i^m X_L Z_L \Pi 
\quad \mbox{and} \quad S_{all} \, Z_L \Pi S_{all}^{-1} =Z_L \Pi.
\]
The assumption that $m$ is odd implies $S^m X S^{-m} =i^m X Z$ and $S^m Z S^{-m} =Z$.
We conclude that   $S_{all}$
implements a gate $S^m$ on the logical qubit which proves Eq.~(\ref{Sall}).

The same arguments as above show that  $H_{all}\Pi H_{all}=\Pi$.
Since $H_{all}$ interchanges the logical operators $X_L$ and $Z_L$,
this proves Eq.~(\ref{Hall}). 
\end{proof}
We note that  $T=T^p S^q Z^r$
whenever the integers $p,q,r$ obey $p+2q+4r=1{\pmod 8}$. 
Using this identity one can implement the logical $T$-gate 
as a composition of the gate $T^p$ with any odd $p$, the
gate $S^q$ with any odd $q$, and the logical Pauli operator $Z_L$.

\section{Commuting Pauli errors through $T$-gates}
\label{sec:Tgate}

Consider a regular code $\css{\calA}{\calB}$
with the logical states $|0_L\rangle$, $|1_L\rangle$ and suppose
the code has a transversal $T$-gate
that can be implemented by a product operator $T_{all}=T^{\otimes n}$.
Let $|\psi\rangle=\alpha|0_L\rangle + \beta|1_L\rangle$ be some 
logical state and  $\rho=|\psi\rangle\langle\psi|$.
In the absence of errors $T_{all}$ preserves the codespace, so that 
$\eta \equiv T_{all} \rho T_{all}^\dag$
is also a logical state. Consider now a memory error $X(e)$.
Our goal is understand what happens when 
$T_{all}$ is applied to a corrupted state $X(e)\rho X(e)$. 
We will show that a composition of $T_{all}$ and a certain Pauli twirling map
 transforms the initial state $X(e)\rho X(e)$
 to a probabilistic  mixture of states $E \eta E^\dag$,
 where $E=X(e)Z(f)$ and  $f\subseteq e$ is random vector drawn from a suitable probability distribution $P(f|e)$.
However, this is true only if the initial error $e$ satisfies a technical condition 
that we call {\em cleanability}. To state this condition, 
represent the binary space $\FF_2^n$ as
a disjoint union of cosets of $\calA$. By definition,  each coset
is a set of vectors $e+\calA$ for some $e\in \FF_2^n$.
Since $\rho$ is stabilized by $X(\calA)$, the state
$X(e)\rho X(e)$ depends only on the coset $e+\calA$.
\begin{dfn}
\label{dfn:clean}
A coset of $\calA$ is cleanable
iff it has a representative $e$ such that 
no vector $g\in \calA^\perp \cap \calO$
has support inside $e$.
\end{dfn}
Recall that the set $\calA^\perp \cap \calO$ describes undetectable
non-trivial $Z$-errors.
Thus a coset is cleanable iff it has a representative whose support contains no   such errors.
In particular, a coset is cleanable whenever it has a representative of 
weight less than $d(\calA)$, see Eq.~(\ref{distance}).
In practice, one can create a list of all cleanable cosets  by examining all vectors $e\in \FF_2^n$ 
and checking whether a set $M(e)\equiv \{ g\in \calO^{|e|}\, : \, g[e]\in \calA^\perp\}$
is non-empty. The  set $M(e)$  is determined by a linear system over $\FF_2$
with $|e|$ variables and $1+\dim{(\calA)}$ equations. If $M(e)$ is empty, the coset
$e+\calA$ is marked as cleanable. Cosets remaining unmarked at the end of this process
are not cleanable. 

From now on we assume that the coset $e+\calA$ is cleanable and
$e$ is a representative that obeys the condition of Definition~\ref{dfn:clean}. 
Define a twirling map $W_\calA$ that applies a random 
$X$-stabilizer drawn from the uniform distribution, 
\[
W_\calA(\omega)=\frac1{|\calA|} \sum_{g\in \calA} X(g) \omega X(g).
\]
Note that $W_\calA$ has trivial action on any logical state. 
Consider states 
\[
\eta=T_{all} \rho T_{all}^\dag \quad \mbox{and} \quad 
\tilde{\eta}=W_\calA( T_{all} X(e) \rho  X(e)  T_{all}^\dag).
\]
The identity
$TXT^\dag = e^{i\pi/4} X S^\dag$
implies
$T_{all} X(e)T_{all}^\dag\sim X(e) S^\dag(e)$,
where $\sim$ stands for some phase factor. Thus
\begin{equation}
\label{rho1}
\tilde{\eta}=W_\calA( X(e)S(e)^\dag \eta S(e) X(e) ).
\end{equation}
Next, the identity $S\sim (I-iZ)/\sqrt{2}$ implies
\[
S^\dag(e)\sim 2^{-|e|/2} \sum_{f\subseteq e} i^{|f|} Z(f),
\]
where $f\subseteq e$ is a shorthand for $\supp{f}\subseteq \supp{e}$.
This yields
\[
X(e) \tilde{\eta} X(e)=2^{-|e|}\sum_{f,f'\subseteq e} i^{|f|-|f'|}  W_\calA(Z(f) \eta Z(f')).
\]
We note that $\eta$ is  invariant under all stabilizers
$X(g)$
that appear in the twirling map since it is a logical state.
 Commuting  a stabilizer $X(g)$ from the twirling map
towards $\eta$ gives an extra phase factor $(-1)^{g(f+f')}$.
Summing up this phase factor over  $g\in \calA$ 
gives a non-zero contribution only if $f+f'\in \calA^\perp$. This shows that 
\begin{equation}
\label{ff'}
X(e)  \tilde{\eta}  X(e)=2^{-|e|} \sum_{\substack{f,f'\subseteq e\\ f+f'\in \calA^\perp\\ }} i^{|f|-|f'|}   Z(f) \eta Z(f').
\end{equation}
By cleanability assumption,
the  sum in Eq.~(\ref{ff'}) gets non-zero contributions only from the terms
with $f+f'\in \calA^\perp \cap \calE\equiv \dot{\calA}$.
Furthermore, $\vphantom{\hat{\hat{A}}} \dot{\calA}=\calB$ since we assumed that the
code is regular.
Define a subspace 
\[
\calB(e)=\{ g\in \calB \, :\, g\subseteq e\}.
\]
Performing a change of variables $f'=f+g$ in Eq.~(\ref{ff'})  gives
\[
X(e)  \tilde{\eta}  X(e) =2^{-|e|} \sum_{f\subseteq e}\; \sum_{g\in \calB(e)} 
i^{|f|-|f+g|} Z(f) \eta Z(f).
\]
Here we noted that $\eta Z(f+g)=\eta Z(f)$
since $Z(g)$ is a stabilizer and $\eta$ is a logical state. 
Next, the identity
$|f+g|=|f|+|g|-2|f\cap g|$  implies  $i^{|f|-|f+g|}=(-1)^{\trn{f} g + |g|/2}$.
Note that $|g|$ is even since $\calB\subseteq \calE$.
We conclude that
\[
X(e)  \tilde{\eta}  X(e) =2^{-|e|}\sum_{f\subseteq e}\;|\calB(e)|^{-1} \sum_{g,h\in \calB(e)} 
(-1)^{\trn{f} g+ \trn{h}g + |g|/2}  Z(f) \eta Z(f).
\]
Here we introduced a dummy variable $h\in \calB(e)$  and used the fact that 
$\eta$ is invariant under stabilizers $Z(h)$ with $h\in \calB(e)$.
The sum over $h$ gives a non-zero contribution only if  $g\in \calB(e)^\perp$. 
We arrive at
\begin{equation}
\label{XrhoX}
X(e) \tilde{\eta}  X(e)= \sum_{f\subseteq e} P(f|e) \, Z(f) \eta Z(f),
\end{equation}
where
\begin{equation}
\label{P(f|e)}
P(f|e)=2^{-|e|} \sum_{g\in \calB(e) \cap \calB(e)^\perp}
(-1)^{\trn{f} g + |g|/2}.
\end{equation}
We claim that $P(f|e)$ is a normalized probability
distribution on the set of vectors $f\subseteq e$,  that is,
\[
\sum_{f\subseteq e} P(f|e)=1 \quad \mbox{and} \quad P(f|e)\ge 0.
\]
Indeed, normalization of $P(f|e)$ follows trivially from Eq.~(\ref{P(f|e)}).
To prove that $P(f|e)\ge 0$ choose an arbitrary basis 
$\calB(e)\cap \calB(e)^\perp=\langle g^1,\ldots,g^m\rangle$.
For any vector $g=\sum_{a=1}^m x_a g^a$ one has
\[
|g|=\sum_{a=1}^m x_a |g^a| -2\sum_{1\le a<b\le m} x_a x_b |g^a\cap g^b| =\sum_{a=1}^m x_a |g^a| {\pmod 4}
\]
since all overlaps $|g^a\cap g^b|$ must be even. It follows that 
\begin{equation}
\label{P(f|e)1}
P(f|e)=2^{-|e|} \prod_{a=1}^m \left[ 1+(-1)^{\trn{f}g^a + |g^a|/2} \right]\ge 0.
\end{equation}

To conclude, the  transversal operator $T_{all}$ followed by the twirling map transforms the initial 
error  $X(e)$  to a probabilistic mixture of
errors $X(e)Z(f)$, where $f\subseteq e$ is drawn from the distribution $P(f|e)$.
However, this is true only if  the coset $e+\calA$ is cleanable.
 More generally, the initial state may have an error $X(e)Z(g)$. Since $Z$-type errors
commute with $T_{all}$, the final error becomes $X(e)Z(g+f)$, where $f\subseteq e$ is a random vector as above.

The above discussion implicitly assumes that the initial error $X(e)$ already includes
a recovery operator applied by the decoder to correct the preceding memory errors.
Ideally, the recovery cancels the error modulo stabilizers. In this case
$X(e)$ itself is a stabilizer and the transversal gate $T_{all}$ creates no additional errors.
If $X(e)$ is not a stabilizer but the coset $e+\calA$ is cleanable, 
application of $T_{all}$ can possibly create new errors $Z(f)$ with $f\subseteq e$.
However, all these new $Z$-errors are either detectable or trivial
due to the cleanability assumption.
Thus, if the decoder 
has a sufficiently small list of candidate initial errors $X(e)$, 
it might be possible to correct the new $Z$-errors at the later stage. 
On the other hand, if the coset $e+\calA$ is not cleanable, 
at least one of the new $Z$-errors is 
undetectable and non-trivial, so the 
the encoded information is lost. 
Although the decoder cannot test cleanability condition, 
this can be easily done in numerical simulations where the actual error
is known at each step. A good strategy is to test the cleanability condition prior to each
application of the logical $T$-gate and abort the simulation whenever the test fails.
Conditioned on passing the cleanability test at each step,
the only effect of the  transversal gates $T_{all}$  is to 
propagate pre-existing $X$-errors to new $Z$-errors as described above.
In this way the simulation can be performed using the standard stabilizer formalism,
even though the circuit may contain  non-Clifford gates.

\section{Maximum likelihood error correction and gauge fixing}
\label{sec:ML}

In this section we develop error correction and gauge fixing algorithms
 for simulation of 
logical Clifford$+T$ circuits in a presence of noise. 
A  fault-tolerant protocol implementing such logical circuit
is modeled as a sequence of the following elementary steps:
\begin{itemize}
\item Memory errors,
\item Noisy syndrome measurements,
\item Code deformations extending the gauge group to a larger group,
\item Code deformations restricting the gauge group to a subgroup,
\item Transversal logical Clifford gates,
\item Transversal logical $T$-gates.
\end{itemize}  
We shall label steps of the protocol by an integer time variable $t=0,1,\ldots,L$.
 At each step $t$ the logical state
is encoded by a  quantum code of CSS type
with a stabilizer group $\calS_t$ and a gauge group $\calG_t$ defined as 
\[
\calS_t= \css{\calA_t}{\calB_t} \quad \mbox{and} \quad
\calG_t=\css{\dot{\calB}_t}{\dot{\calA}_t}.
\]
Each code has one logical qubit
with logical Pauli operators $X_L=X(\overline{1})$, $Z_L=Z(\overline{1})$
 and $n$ physical qubits. 
A state of the physical qubits at the $t$-th step
will be represented as $E_t \omega_t E_t^\dag$,
where $\omega_t$ is some logical state of the code $\calS_t$
and  $E_t$ is a Pauli  error.
To simplify notations, we shall represent 
$n$-qubit Pauli operators $X(a)Z(b)$  by binary vectors $e\in \FF_2^{2n}$ such that
$e=(a,b)$. Here we ignore the overall phase.
Multiplication in the Pauli group corresponds to addition in $\FF_2^{2n}$.
Each Pauli error $e\in \FF_2^{2n}$ is contained in some coset of the gauge group $\calG_t$,
namely, $e+\calG_t$.   The decoder described below operates on the level of cosets without ever trying to distinguish
errors from the same coset. This is justified since 
the logical state $\omega_t$ is invariant under gauge operators,
see Section~\ref{sec:subs}.
We shall parameterize cosets of $\calG_t$ by binary vectors as follows.
Let $A_t$ and $B_t$ be some fixed generating matrices of subspaces $\calA_t+\langle \overline{1}\rangle$ and 
$\calB_t+\langle \overline{1}\rangle$. Define a matrix
\begin{equation}
\label{Ct}
C_t=\left[ \ba{cc} B_t  &  \\  & A_t \\ \ea \right]. 
\end{equation}
Then the coset of a Pauli error $e=(a,b)$ is parameterized by a vector $f=C_te=(B_ta,A_t b)$.
Note that  the matrix $C_t$ has size $c_t\times 2n$, where $c_t=2+\dim{(\calS_t)}$.
Accordingly, the gauge group $\calG_t$ has $2^{c_t}$ cosets.
Furthermore, if $\css{\calA_t}{\calB_t}$ is a regular code, that is,
$\calA_t=\dot{\calB_t}$  and $\calB_t=\dot{\calA}_t$ 
then $A_t$ and $B_t$ become generating matrices of $\calB_t^\perp$
and $\calA_t^\perp$ respectively. In this case 
$B_ta$ determines the coset of $\calA_t$ that contains $a$
while $A_tb$ determines the coset of $\calB_t$ that contains $b$.
Furthermore, $c_t\le n+1$ with the equality in the case of regular codes.

We assume that  memory errors occur at half-integer times
$t-1/2$, where $t=1,\ldots,L$.
A memory error is modeled  by a random vector $e\in \FF_2^{2n}$
drawn from some fixed probability distribution $\pi$. For example, 
the depolarizing noise that independently applies one of the Pauli 
operators $X,Y,Z$ to each qubit with probability $p/3$ 
can be described by 
$\pi(e)=\prod_{j=1}^n \pi_1(e_j,e_{n+j})$, where 
$\pi_1(0,0)=1-p$ and $\pi_1(1,0)=\pi_1(0,1)=\pi_1(1,1)=p/3$.
Let $e^t\in \FF_2^{2n}$ be the memory error that occurred at 
time $t-1/2$
and $e^1+\ldots+e^t$ be the accumulated error at the $t$-th step.
Assuming that  $e\in \FF_2^{2n}$ is drawn from the distribution $\pi$,
the corresponding coset $f=C_te$ is drawn from a distribution
\begin{equation}
\label{P_t}
P_t(f)=\sum_{e\, : \, C_te=f} \pi(e).
\end{equation}
The coset of $\calG_t$ that contains the  accumulated error  
at the $t$-th step is 
\begin{equation}
\label{full}
f^t=C_t(e^1+e^2+\ldots+e^t).
\end{equation}

Syndrome measurements occur at integer times $t=1,2,\ldots,L$.
The measurement performed at the $t$-th step determines syndromes for
some subset of $b_t$ stabilizers from the stabilizer group $\calS_t$.
These stabilizers may or may not be independent.
They may or may not generate the full  group $\calS_t$.
The syndrome measurement  reveals a partial information about the coset 
of the accumulated error $f^t$ which can be represented by a vector 
\[
s^t=M_t f^t\in \FF_2^{b_t}
\]
for some binary matrix $M_t$ of size $b_t \times c_t$.
We shall refer to the vector $s^t$ as a {\em syndrome}.
A precise form of the matrix $M_t$  is not important at this point. 
The only restriction that we impose is that  $M_tC_t e=0$ if  $e$ represents a logical Pauli operator,
that is, $e=(\overline{1},\overline{0})$ or $e=(\overline{0},\overline{1})$.
This ensures that syndrome measurements do not affect the logical qubit.
In a presence of measurement errors the observed syndrome 
may differ from the actual one. 
We model a measurement error
by  a random vector $e'\in \FF_2^{b_t}$ drawn from some fixed probability distribution $Q_t$
such that the syndrome observed at the $t$-th step
is $s^t=M_tf^t + e'$.
The family of  matrices $C_t,M_t$ and the probability distributions $P_t,Q_t$
capture all features of the error model that matter for the decoder.
The decoding problem can be stated as follows.

\vspace{3mm}
\noindent
{\bf ML Decoding:} {\em Given a list of observed syndromes $s^1,\ldots,s^L$,
determine which coset of the gauge group $\calG_L$ is most likely to contain the
accumulated error at the last  time step.}
\vspace{3mm}

Below we describe an implementation of the ML decoder with  the running time $O(L c2^c)$,
where $c=\max_t c_t$.
The decoder can be used in the online regime such that
the syndrome $s^t$ is revealed only at the time step $t$.
 In this setting the running time is $O(c2^c)$ per time step. 
The decoder also includes a preprocessing step that 
computes  matrices $C_t,M_t$, the effective error model $P_t,Q_t$,
and certain additional data required for implementation of logical $T$-gates. 
The preprocessing step takes time $2^{O(n)}$ in the worst case. 
Since this step  does not depend on the measured syndromes, it  can be performed offline.
The preprocessing can be done more efficiently if the 
code has a special structure. 

Consider first the simplest case when all codes are the same
and no logical gates are applied. Our model then describes a quantum memory
with one logical qubit. Since all codes are the same, we temporarily suppress the subscript $t$
in $\calG_t,C_t,M_t,P_t,Q_t$ and $c_t,b_t$. 
For simplicity, we assume that the initial state at time $t=0$ has no errors, that is,
$f^0=\overline{0}$. 
Consider some time step $t$ and  let $f\in \FF_2^c$ be a coset of $\calG$.
Let $\rho_t(f)$ be the probability that $f$ contains the accumulated  error
at the $t$-th step, $f=C(e^1+\ldots+e^t)$, 
conditioned on the observed syndromes $s^1,\ldots,s^t$.
Since we assume that all errors occur independently, one has 
\begin{equation}
\label{eq111}
\rho_t(f)=\sum_{f^1,\ldots,f^{t-1}\in \FF_2^c}\; \;  \prod_{u=1}^t P(f^u+f^{u-1}) \prod_{u=1}^t Q(s^u + Mf^u),
\end{equation}
where we set $f^t\equiv f$ and $f^0\equiv \overline{0}$.
Also we ignore the overall normalization of $\rho_t(f)$
which depends only on the observed syndromes. 
We set $\rho_0(f)=1$ if $f=\overline{0}$ and $\rho_0(f)=0$ otherwise.
It will be convenient to represent the probability distribution $\rho_t(f)$ by a 
{\em likelihood vector}
\begin{equation}
\label{likelihood}
|\rho_t\rangle = \sum_{f\in \FF_2^{c}} \rho_t(f)\, |f\rangle
\end{equation}
that belongs to the  Hilbert space of $c$ qubits. 
Define $c$-qubit operators 
\begin{equation}
\label{PQ}
P=\sum_{f,g\in \FF_2^{c}} P(f+g) |f\rangle\langle g|
\quad \mbox{and} \quad
Q_t=\sum_{f\in \FF_2^c} Q(s^t + Mf) |f\rangle\langle f|.
\end{equation}
Then
\[
|\rho_t\rangle = (Q_t P) \cdots  (Q_2 P) ( Q_1 P)|\overline{0}\rangle.
\]
Next we observe that $P$ commutes with any $X$-type Pauli operator,
\[
\langle f| X(h) P|g\rangle =
\langle f+h|P|g\rangle = P(f+g+h) = \langle f|P|g+h\rangle= \langle f|P  X(h) |g\rangle.
\]
Thus $P$ is diagonal in the $X$-basis
$\{ H|f\rangle\}$, where
 $H$ is the Walsh-Hadamard transform on $c$ qubits, 
\[
H|f\rangle =2^{-c/2} \sum_{g\in \FF_2^{c}} (-1)^{\trn{f} g} |g\rangle.
\]
Note that $H^2=I$. 
The above shows that a  matrix
$\hat{P} \equiv HPH$ is diagonal in the $Z$-basis,
\begin{equation}
\label{chi}
\hat{P} |f\rangle = \hat{P}(f) |f\rangle, \quad \mbox{where} \quad \hat{P}(f)=\sum_{g\in \FF_2^{c}} (-1)^{\trn{f} g} P(g).
\end{equation}
Assuming that the likelihood vector $\rho_{t-1}$ has been already computed
and stored in a classical memory as a real vector of size $2^c$, one can compute
$\rho_t$ in four steps:
\begin{enumerate}
\item Apply the Walsh-Hadamard transform $H$.
\item Apply the diagonal matrix $\hat{P}$.
\item Apply the Walsh-Hadamard transform $H$.
\item Apply the diagonal matrix $Q_t$.
\end{enumerate}
Obviously, steps (2,4) take time $O(2^{c})$. Using the fast Walsh-Hadamard transform
one can implement steps $(1,3)$ in time $O(c2^{c})$. 
Overall, computing all  likelihood vectors $\rho_1,\ldots,\rho_L$ requires time 
$O(Lc2^c)$.  Finally, the most likely coset of $\calG$
at the final time step is 
$m^L\equiv \arg \max_f \rho_L(f)$. It can be found in time $O(2^c)$. 
If one additionally assumes that the final syndrome measurement is noiseless,
there are only four cosets consistent with the final syndrome $s^L$. In this case the
most likely coset $m^L$ can be computed in time $O(1)$.

 Less formally, the ML decoding can be viewed as a competition between two 
opposing forces. Memory errors described by $P$-matrices
tend to spread the support of the likelihood vector over many cosets, whereas
syndrome measurements described by $Q$-matrices tend to localize the likelihood vector on  cosets $f$
whose syndrome $Mf$ is sufficiently close to the observed one.
Typically, the most likely coset $m^L$ coincides with the coset of the accumulated error
if the likelihood vector $\rho_t$ remains sufficiently peaked at all time steps.
Once the likelihood vector becomes flat, a successful decoding might not be possible
and the encoded information is lost.

Next let us  extend the ML decoder to code deformations that  can
change the gauge group $\calG_t$, see Section~\ref{sec:subs}.
For concreteness, assume that  a code deformation
mapping $\calG_{t-1}$ to $\calG_t$
 occurs at time $t-1/4$.
We assume that for any  time step $t$ one has
either $\calG_{t-1}\subseteq \calG_t$ or $\calG_{t-1}\supseteq \calG_t$.
The general case can be reduced  to one of these cases by adding a  dummy 
time step with a gauge group $\calG_{t-1}\cap \calG_t$,
such that first $\calG_{t-1}$ is reduced to $\calG_{t-1}\cap \calG_t$
and next $\calG_{t-1}\cap \calG_t$  is extended to $\calG_t$.

\noindent
{\em Case~1:} $\calG_{t-1}\subseteq \calG_t$.
Then each coset of $\calG_{t-1}$ uniquely determines a coset of $\calG_t$.
In other words,  $f^t=D_t f^{t-1}$ for some binary matrix $D_t$
of size $c_t\times c_{t-1}$. Note that $c_t\le c_{t-1}$ since a larger group
has fewer cosets.
Then the code deformation is equivalent to updating the likelihood
vector according to 
\begin{equation}
\label{merge}
|\rho_t\rangle \gets \sum_{f\in \FF_2^{c_{t-1}}} \; \rho_t(f) |D_t f\rangle.
\end{equation}
Let us show that the updated vector can be computed 
 in time 
$O(2^{c_{t-1}})$. Indeed, for any fixed $g\in \FF_2^{c_t}$ 
a linear system $D_t f=g$ has $2^{c_{t-1}-c_t}$ solutions $f\in \FF_2^{c_{t-1}}$.
One can create a list of all solutions $f$
and sum up the amplitudes $\rho_t(f)$  in 
time $O(2^{c_{t-1}-c_t})$. This would give a single amplitude 
of the updated likelihood vector. 
Since the number of such amplitudes is $2^{c_t}$, 
the overall time scales as 
 $O(2^{c_{t-1}})$. 
A good strategy for numerical simulations is to choose 
matrices $C_t$ in Eq.~(\ref{Ct})  such that 
$C_t$ is obtained from $C_{t-1}$ by removing a certain subset of $c_{t-1}-c_t$ rows.
Then $f^t$ is obtained from $f^{t-1}$ by erasing a certain subset  of bits
and the  update Eq.~(\ref{merge})  amounts to 
taking a partial trace.

\noindent
{\em Case~2:} $\calG_{t-1}\supseteq \calG_t$.
 Note that $c_t\ge c_{t-1}$ since a smaller group
has more cosets.  
Then each coset of $\calG_{t-1}$ splits into $2^{c_t-c_{t-1}}$ cosets of $\calG_t$.
Thus  $f^{t-1}=D_t f^t$ for some binary matrix $D_t$ of size $c_{t-1}\times c_t$.
As we argued in Section~\ref{sec:subs}, the coset $f^t$ is drawn from the uniform
distribution on the set of all cosets of $\calG_t$ contained in $f^{t-1}$,
see Eq.~(\ref{rhoL2}).
Then the code deformation is equivalent to updating the likelihood
vector according to 
\begin{equation}
\label{split}
|\rho_t\rangle \gets2^{c_{t-1}-c_t}  \sum_{f\in \FF_2^{c_t}} \; \rho_t(D_t f) |f\rangle.
\end{equation}
The same arguments as above show that the
update can be performed in time $O(2^{c_t})$.
A good strategy for numerical simulations is to choose
matrices $C_t$ in Eq.~(\ref{Ct})  such that 
$C_t$ is obtained from $C_{t-1}$ by adding  a certain subset of $c_t-c_{t-1}$ rows.
Then $f^t$ is obtained from $f^{t-1}$ by inserting random bits on a subset
of coordinates 
and the  action of $\tilde{D}_t$ on the likelihood vector amounts to 
taking a tensor product with the maximally mixed state of $c_t-c_{t-1}$ qubits.

Next let us discuss transversal logical gates. 
For concreteness, assume that a logical gate applied at a step $t$
occurs at time $t+1/4$, where $t=0,\ldots,L-1$.
A logical gate can be applied at a step $t$
only if the corresponding  code is regular, that is, 
$\calS_t=\calG_t=\css{\calA_t}{{\calB}_t}$,
where $\calB_t=\dot{\calA}_t$.
We only allow logical gates from the Clifford$+T$ basis
and assume that the code $\calS_t$ 
obeys  transversality conditions of Lemma~\ref{lemma:transversal}.
Let $V$ be one of the transversal operators
$H_{all}$, $S_{all}$, $T_{all}$ defined in Section~\ref{sec:trans}.

We start from logical Clifford gates. 
Let $P_t$ be the Pauli operator that represents the accumulated 
error $e^1+\ldots+e^t$.
The initial state before application of $V$
has a form $P_t \omega_t P_t^\dag$,
where $\omega_t$ is a logical state
of the code $\calS_t$.
The final state after application of $V$ is
\[
VP_t \omega_t P_t^\dag V^\dag = Q_t \tilde{\omega}_t Q_t^\dag
\]
where  $\tilde{\omega}_t=V\omega_t V^\dag$ is  a new logical state
 and $Q_t=VP_t V^\dag$ is a new
Pauli error. Thus the likelihood of a coset $P\calG_t$ before application of $V$
must be the same as the likelihood of a coset $Q \calG_t$ after application of $V$,
where $P$ is an arbitrary Pauli error and $Q=VPV^\dag$.
Note that the coset $Q\calG_t$ 
depends only on the coset of $P$ since
$V\calG_tV^\dag=\calG_t$, see  the proof of  Lemma~\ref{lemma:transversal}.
Thus the action of $V$ on cosets of $\calG_t$ defined above can be described by
a linear map $v\,: \, \FF_2^{c_t} \to \FF_2^{c_t}$.
For concreteness, assume that  cosets $f\in \FF_2^{c_t}$ are  parameterized as in Eq.~(\ref{Ct})
such that an error $X(a)Z(b)$ belongs to a coset
$f=(\alpha,\beta)$, where
$\alpha=B_ta$ and $\beta=A_tb$.
If $V$ implements a logical $H$-gate then
$v(\alpha,\beta)=(\beta,\alpha)$. If $V$ implements a logical $S$-gate
then $v(\alpha,\beta)=(\alpha,\alpha+\beta)$. 
One can compute the action of any other Clifford gate on cosets
in a similar fashion. 
This shows that application of $V$ is equivalent to updating the likelihood vector
according to 
\begin{equation}
\label{v}
|\rho_t\rangle \gets \sum_{f\in \FF_2^{c_t}} \rho_t(f) |vf\rangle.
\end{equation}
This update can be performed in time  $O(2^{c_t})$.

Next assume that $V=T_{all}$ is a   transversal $T$-gate. 
As we argued in Section~\ref{sec:Tgate}, it is desirable
to correct pre-existing memory errors
of $X$-type before applying the $T$-gate.
The ML decoder performs error correction
by examining the current likelihood vector $\rho_t$
and choosing a  recovery operator from the most likely coset
of $X$-errors
\begin{equation}
\label{recovery}
\alpha^*=\arg \max_{\alpha} \sum_{\beta} \rho_t(\alpha,\beta).
\end{equation}
Applying a recovery operator $X(\alpha^*)$ 
is equivalent to updating the likelihood vector according to
\[
|\rho_t\rangle \gets \sum_{\alpha,\beta} \rho_t(\alpha,\beta)|\alpha+\alpha^*,\beta\rangle.
\]
This update can be performed in time $O(2^{c_t})$.

Recall that for regular CSS codes a Pauli error $(a,b)$
belongs to a coset $(\alpha,\beta)=C_t(a,b)$, where $\alpha=B_ta$ is the coset of $\calA_t$
that contains $a$ and $\beta=A_tb$ is the coset of $\calB_t$ that contains $b$.
Let $\calN_t$ be a set of cleanable cosets of $\calA_t$,
see Definition~\ref{dfn:clean}.
As we argued in Section~\ref{sec:Tgate}, a lookup table of cleanable
cosets can be computed in time $O(2^n)$ at the preprocessing step. 
The next step of the ML decoder is to project the likelihood vector 
on the subspace spanned by cleanable cosets,
\[
|\rho_t\rangle \gets \sum_{\alpha \in \calN_t}\sum_\beta  \rho_t(\alpha,\beta ) |\alpha,\beta\rangle.
\]
This update can be performed in time $O(2^{c_t})$.

Now the decoder is ready to apply the transversal gate $T_{all}$.
Here we assume for simplicity that $T_{all}=T^{\otimes n}$.
For each cleanable coset $\alpha$  let $e(\alpha)\in \FF_2^n$
be some fixed representative of $\alpha$ that satisfies
condition of Definition~\ref{dfn:clean}, that is, 
$B_te(\alpha)=\alpha$ and 
the support of $e(\alpha)$ contains no vectors from $\calA_t^\perp \cap \calO$.
A lookup table of vectors $e(\alpha)$ 
must be  computed at the preprocessing step.
As was argued in Section~\ref{sec:Tgate}, a composition of $T_{all}$ and 
a twirling map $\calW_{\calA_t}$ maps a pre-existing Pauli error
$X(e)Z(g)$ to a probabilistic mixture of Pauli errors
$X(e)Z(f+g)$, where $f\subseteq e$ is a random vector drawn
from distribution $P(f|e)$ defined in Eqs.~(\ref{P(f|e)},\ref{P(f|e)1}).
Since the twirling map has no effect on the likelihood vector, it can be skipped.
Thus $T_{all}$  updates the likelihood vector according to
\begin{equation}
\label{Tupdate1}
|\rho_t\rangle \gets
\sum_{\alpha,\beta}\;  \sum_{f\subseteq e(\alpha)} P(f|e(\alpha)) \rho_t(\alpha,\beta) |\alpha,\beta+A_t f\rangle.
\end{equation}
Define an auxiliary operator $\Gamma_\alpha$ such that 
\begin{equation}
\label{Gamma}
\Gamma_\alpha |\beta\rangle =\sum_{f\subseteq e(\alpha)} P(f|e(\alpha)) |\beta+A_t f\rangle.
\end{equation}
This operator acts on the Hilbert space 
of $m$ qubits, where $m=\dim{(\calB_t^\perp)}=1+\dim{(\calA_t)}$.
Then Eq.~(\ref{Tupdate1}) can be written as  
\begin{equation}
\label{Tupdate2}
|\rho_t\rangle \gets \left(\sum_\alpha |\alpha\rangle\langle \alpha| \otimes \Gamma_\alpha\right) |\rho_t\rangle.
\end{equation}
We note that  $\Gamma_\alpha$ is diagonal in the $X$-basis since
\[
\Gamma_\alpha=\sum_{f\subseteq e(\alpha)} P(f|e(\alpha))  X(A_t f).
\]
Let $H$ be  the Walsh-Hadamard transform  on $m$ qubits.
Define an  operator
 $\hat{\Gamma}_\alpha=H\Gamma_\alpha H$ which is diagonal in the $Z$-basis
of $m$ qubits. The 
 diagonal matrix elements of $\hat{\Gamma}_\alpha$ are
\begin{equation}
\label{Tupdate3}
\langle \beta|\hat{\Gamma}_\alpha|\beta\rangle
=\sum_{f\subseteq e(\alpha)} P(f|e(\alpha)) (-1)^{\trn{\beta} A_t f}.
\end{equation}
Below we consider some fixed $\alpha$ and use a shorthand $e\equiv e(\alpha)$.
Substituting the explicit expression for $P(f|e)$ from Eq.~(\ref{P(f|e)})
into Eq.~(\ref{Tupdate3}) yields
\begin{equation}
\label{Tupdate4}
\langle \beta|\hat{\Gamma}_\alpha|\beta\rangle
=\sum_{g\in \calB_t(e)\cap \calB_t(e)^\perp}\;
2^{-|e|} \sum_{f\subseteq e} 
(-1)^{\trn{f} g  + |g|/2 + \trn{\beta}A_t f}.
\end{equation}
The sum over $f$ vanishes unless the restriction of the vector
$\trn{A}_t\beta$ onto $e$ coincides with $g$.
Define a diagonal $n\times n$ matrix $J_e$ that has non-zero
entries in the support of $e$, that is, $\mathrm{diag}(J_e)=e$. Then 
we get a constraint $g=J_e\trn{A}_t\beta$
and $g\in \calB_t(e)\cap \calB_t(e)^\perp$.
This shows that 
\begin{equation}
\label{Tupdate5}
\langle \beta|\hat{\Gamma}_\alpha|\beta\rangle
=\left\{
\ba{rcl}
(-1)^{\frac12 |J_e \trn{A}_t \beta |} &\mbox{if} & J_e \trn{A}_t \beta \subseteq \calB_t(e)\cap \calB_t(e)^\perp, \\
0 && \mbox{otherwise}. \\
\ea 
\right.
\end{equation}
The action of the  diagonal operator $\hat{\Gamma}_\alpha$ on any
state of $m$ qubits can be computed in time $O(2^m)$
using Eq.~(\ref{Tupdate5}).
Using the fast Walsh-Hadamard transform on $m$ qubits
one can compute the action of the original  operator $\Gamma_\alpha$
defined in Eq.~(\ref{Gamma}) on any
state of $m$ qubits in time $O(m2^m)$. 
The number of cleanable cosets $\alpha$ is upper bounded
by the total number of cosets of $\calA_t$.
The latter is equal to $2^k$, where $k=\dim{(\calA_t^\perp)}=1+\dim{(\calB_t)}$.
The updated likelihood vector can be computed  by evaluating each
term $\alpha$ in Eq.~(\ref{Tupdate2}) individually and summing up these terms.
This calculation requires time $O(m2^{k+m})=O(c_t2^{c_t})$.
Overall, application of the transversal $T$-gate requires time $O(c_t2^{c_t})$.

\section{Logical Clifford$+T$ circuits  with the $15$-qubit code}
\label{sec:CT15}

In this section we apply the ML decoder 
to a specific fault-tolerant protocol which similar to the one proposed by  Anderson et al~\cite{Anderson2014}.
The protocol simulates a logical quantum circuit composed of Clifford and $T$ gates
by alternating between two error correcting codes: the 
$7$-qubit color code (a.k.a. the Steane code) and the
$15$-qubit Reed-Muller code.
We shall label these two codes $C$ and $T$ 
since they provide  transversal Clifford gates and the $T$-gate respectively.
Both $C$ and $T$ are regular CSS codes that can be obtained by the gauge fixing method  from 
a single subsystem CSS code that we call the base code.
This section
may also serve  as simple example  illustrating the general construction 
of doubled color codes
described in Sections~\ref{sec:summary}-\ref{sec:gadget}.

Our starting point is the $7$-qubit color code.
Consider a graph $\Lambda=(V,E)$ 
with a set of vertices $V=[7]$ and a set of edges $E$
shown on Fig.~\ref{fig:steane}(a).
For consistency with the subsequent sections, we shall refer to the graph
$\Lambda$ as a lattice. Vertices of $\Lambda$ will be referred to as sites. 
We place one qubit at each site of the lattice.
Let $f^1,f^2,f^3\subseteq \FF_2^7$ be the three faces of $\Lambda$, 
see Fig.~\ref{fig:steane}(b), and $\calS\subseteq \FF_2^7$ be the three-dimensional
subspace spanned by the faces, $\calS=\langle f^1,f^2,f^3\rangle$.
The faces $f^1,f^2,f^3$ can also be viewed as rows of a parity check matrix 
\[
  \begin{array}{|c|c|c|c|c|c|c|}  \hline
    1 &   & 1  &	  & 1 &  & 1 \\  \hline
      & 1 &  1 &  &   & 1 & 1 \\  \hline
      &   &  & 1 & 1 &  1  & 1 \\  \hline
  \end{array}
\]
The $7$-qubit color code has a stabilizer group $\css{\calS}{\calS}$.
This is a regular CSS code with six stabilizers $X(f^i)$, $Z(f^i)$.
The code distance is $d=d(\calS)=3$. 
Minimum weight logical operators  are $X(\omega)$, $Z(\omega)$, where 
\[
\omega=\overline{1}+f^2=e^1+e^4+e^5.
\]
Recall that $e^i$ denotes the standard basis vector of the binary space.
The subspace $\calS$ is doubly even,
$|f|=0{\pmod 4}$ for all $f\in \calS$. Lemma~\ref{lemma:transversal} implies
that the color code has transversal logical gates $H$ and $S$
that can be realized by operators $H_{all}=H^{\otimes 7}$ and $S_{all}=S^{\otimes 7}$.

\begin{figure}[h]
\centerline{\includegraphics[height=4cm]{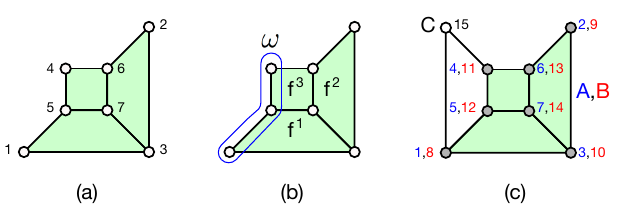}}
\caption{(a) Color code lattice $\Lambda$.  The $7$-qubit color code is defined by placing one qubit
at each site of $\Lambda$.
(b) Faces $f^1,f^2,f^3$ and the subset $\omega$
define stabilizers $X(f^i)$, $Z(f^i)$ and logical operators 
$X(\omega)$, $Z(\omega)$ respectively.
(c) A doubled color code  is obtained  by making two copies of $\Lambda$
labeled $A$ and $B$ placed atop of each other,  and adding one additional qubit labeled $C$. 
\label{fig:steane}
}
\end{figure}

Suppose now that each site of the lattice $\Lambda$ contains two  qubits
labeled $A$ and $B$.  
Let us add one additional  qubit labeled $C$. We assume that  $C$ is placed next to the lattice
$\Lambda$ such that $C$ and $\omega$ form one additional face, see  Fig.~\ref{fig:steane}(c).
This defines a system of $15$ qubits that are partitioned into three consecutive
blocks, $[15]=ABC$, where $|A|=|B|=7$ and $|C|=1$.

The $T$-code that provides a transversal  $T$-gate
has a stabilizer group  $\css{\calT}{\dot{\calT}}$,
 where $\calT\subseteq \FF_2^{15}$ is a
four-dimensional subspace spanned by  double faces of the lattice $\Lambda$
and by the all-ones vector supported on the region $BC$,
\begin{equation}
\label{Tspace}
\calT=\langle f^i[A]+f^i[B]\rangle + \langle BC\rangle.
\end{equation}
Here  $i=1,2,3$ and $BC\equiv \overline{1}[BC]$. 
One can also view
basis vectors of $\calT$ as rows of a parity check matrix 
\[
  \begin{array}{|c|c|c|c|c|c|c|c|c|c|c|c|c|c|c|}  \hline
    1 &   & 1  &	  & 1 &    & 1 &  &1   &   & 1   & & 1 &  & 1 \\  \hline
      & 1 &  1 &    &     & 1 & 1 &  &     & 1 &  1 &  &   & 1 & 1  \\  \hline
      &   &  & 1 & 1 &  1  & 1 &  &   &   &  & 1 & 1 &  1  & 1  \\  \hline
      &   &  &    &    &      &    &  1 & 1  & 1  &  1 &  1  &    1 &  1    &   1  \\ \hline
  \end{array}
\]
Here the first three rows stand for $f^i[A]+f^i[B]$
and the last row stands for $BC$. 
A direct inspection shows that $\calT$ is  triply even, $|f|=0{\pmod 8}$ for all $f\in \calT$.
By Lemma~\ref{lemma:transversal}, the $T$-code has
a transversal logical $T$-gate realized by an operator $T_{all}=T^{\otimes 15}$. 
Let us explicitly describe $Z$-stabilizers of the $T$-code.
Define a subspace $\calG\subseteq \FF_2^{15}$ spanned by 
double edges of $\Lambda$, 
\begin{equation}
\label{Gedge}
\calG=\langle l[A]+l[B]\, : \, l\in E\rangle.
\end{equation}
For example, a double edge $l=(2,3)\in E$ gives rise to a basis vector
$l[A]+l[B]=e^2+e^3+e^9+e^{10}$, see Fig.~\ref{fig:steane}(c).
Clearly, double edges $l[A]+l[B]$ have an even overlap with double faces
$f^i[A]+f^i[B]$ 
as well as with the vector $BC$. This shows that 
$\calG\subseteq  \dot{\calT}$. 
Define also a subspace $\calC\subseteq \FF_2^{15}$ spanned by
single faces of $\Lambda$, including the extra face formed by $\omega$ and $C$,
namely
\begin{equation}
\label{Cspace}
\calC=\langle f^i[A]\rangle + \langle f^i[B]\rangle + \langle \omega[B]+C\rangle,
\end{equation}
where $i=1,2,3$ and $C\equiv 1[C]$.
A direct inspection shows that 
$\calC\subseteq \dot{\calT}$ and, moreover,
\begin{equation}
\label{Gspace}
\dot{\calT}=\calC+\calG.
\end{equation}
To summarize,  $Z$-stabilizers of the $T$-code fall into three classes:
(i) {\em edge-type}  stabilizers $Z(f)$, where $f=l[A]+l[B]$ is a double edge of the color code lattice
$\Lambda$, (ii) {\em face-type} stabilizers $Z(g)$, where  $g=f^i[A]$
 or $g=f^i[B]$ is a face of the lattice $\Lambda$,
and (iii) a special face-type stabilizer  
$g=\omega[B]+C$ that represents the 
extra face connecting $C$ and $\omega$. 
These stabilizers are not independent. For example, the product
of edge-type stabilizers over any closed loop on the color code lattice
$\Lambda$ gives the identity. 
Minimum weight logical operators of the $T$-code
can be chosen as $X(A)$ and $Z(\omega[A])$.
The code distance is $d=3$
since $d(\calT)=3$ and $d(\dot{\calT})=7$.

The $C$-code that provides   transversal Clifford gates has a stabilizer group
$\css{\calC}{\calC}$, where $\calC$ is defined by Eq.~(\ref{Cspace}). 
Stabilizers of this code can be partitioned into three classes:
(i) stabilizers $X(f^i[A])$, $Z(f^i[A])$  define the $7$-qubit color code on the
region $A$ with logical operators $X(\omega[A])$, $Z(\omega[A])$, 
(ii) stabilizers  $X(f^i[B])$, $Z(f^i[B])$ define the $7$-qubit color code
on the region $B$ with logical operators $X(\omega[B])$, $Z(\omega[B])$, and (c) stabilizers
 $X(\omega[B]+C)$,
$Z(\omega[B]+C)$ define the two-qubit EPR state $|0,0\rangle + |1,1\rangle$
shared between $B$ and $C$ such that the first qubit of the EPR state is encoded
by the $7$-qubit color code. 
Thus any logical state of the $C$-code has a form
\begin{equation}
\label{psi_L}
|\psi_L\rangle = (\alpha |0_L\rangle + \beta |1_L\rangle)_A \otimes (|0_L 0\rangle + |1_L 1\rangle)_{BC},
\end{equation}
where $\alpha,\beta\in \CC$ are some coefficients and
$|0_L\rangle$, $|1_L\rangle$ are the  logical basis states of the $7$-qubit  color code.
By discarding qubits of $BC$ one can convert the $C$-code   into the $7$-qubit color code.
Thus these two codes have the same logical operators, the same distance, and the same
transversality properties. In particular, the $C$-code 
provides a transversal implementation of the full Clifford group.
This can also be seen from Lemma~\ref{lemma:transversal} by noting that 
$\dot{\calC}=\calC$  and 
$\calC$ is doubly even,
$|f|=0{\pmod 4}$ for all $f\in \calC$.

The base code appears as an intermediate step
in the conversion between $C$ and $T$ codes. 
The stabilizer group of the base code is defined as  the intersection of 
 stabilizer groups $\css{\calC}{\calC}$  and $\css{\calT}{\dot{\calT}}$
describing  the $C$ and $T$ codes.
Thus the base code has a stabilizer group $\css{\calC\cap \calT}{\calC\cap \dot{\calT}}$.
By definition, $\calC\subseteq \dot{\calT}$, see Eq.~(\ref{Gspace}),
that is, $\calC\cap \dot{\calT}=\calC$. 
We claim that $\calT\subseteq \calC$.
Indeed,  all double face generators of $\calT$ are contained in $\calC$
by definition, 
so we just need to check that $BC\in \calC$.
Using the identity
$f^2+\omega=\overline{1}$, see Fig.~\ref{fig:steane}(a,b), one gets
\begin{equation}
\label{1BC}
BC=B+C=(f^2+\omega)[B]+C=f^2[B]+(\omega[B]+C)\in \calC.
\end{equation}
We conclude that the base code has a stabilizer group $\css{\calT}{\calC}$
and the gauge group $\css{\dot{\calC}}{\dot{\calT}}=\css{\calC}{\dot{\calT}}$.
The base code has distance $d=3$ since $d(\calT)=d(\calC)=3$.
All three codes have the same logical operators defined in Eq.~(\ref{XYZL}).
We summarize definitions of the three codes in Table~\ref{table:codes}. 
\begin{table}[!ht]
\centerline{
\begin{tabular}{r|c|c|c|c|c}
 & Transversal gates  & $X$-stabilizers & $Z$-stabilizers & stabilizer group  & gauge group\\
\hline
$\vphantom{\hat{\hat{A}}}$
$C$-code & Clifford group &  single faces & single faces & $\css{\calC}{\calC}$  &  $\css{\calC}{\calC}$   \\
\hline
$\vphantom{\hat{\hat{A}}}$
$T$-code & $T$ gate &  double faces&
single faces& $\css{\calT}{\dot{\calT}}$  & $\css{\calT}{\dot{\calT}}$  \\
& & $BC$ & double edges & &  \\
\hline
$\vphantom{\hat{\hat{A}}}$
Base code &   & double faces & single faces & $\css{\calT}{\calC}$   & $\css{\calC}{\dot{\calT}}$ \\
 & & $BC$ &  & & \\
\hline
\end{tabular}}
\caption{The family of codes used in the protocol.}
\label{table:codes}
\end{table}

\begin{figure}[ht]
\centerline{\includegraphics[height=5.5cm]{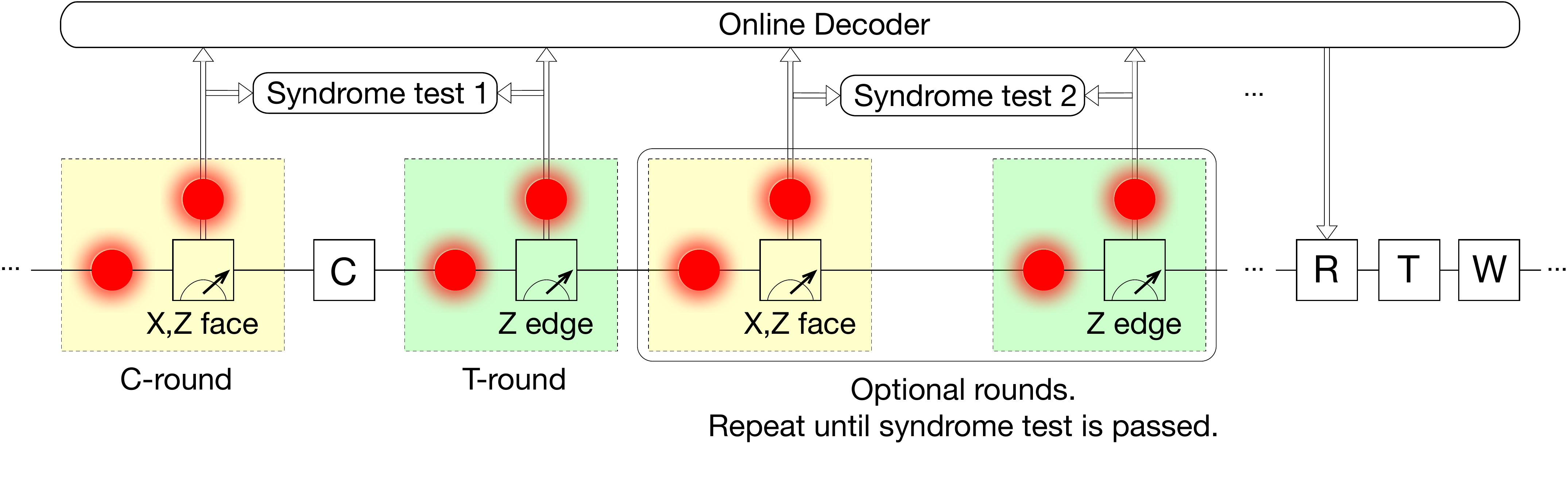}}
\caption{A fault-tolerant implementation of a  Clifford$+T$ circuit.
Red circles  represent depolarizing memory errors and syndrome measurement errors.
Transversal Clifford and $T$-gates are shown by $C$ and $T$ boxes.
The boxes $W$ and $R$ represent  the  twirling map 
and the recovery operator respectively.
Measurement boxes represent (partial) syndrome measurements.
Face-type stabilizers of the $C$-code
 are measured in every $C$-round.
Edge-type stabilizers of the $T$-code are measured in every $T$-round.
Optional rounds are added only if the 
full syndrome of the $T$-code inferred from these measurements
fails to pass a consistency test. 
The rounds continue in the periodic fashion. 
On average, the protocol applies one logical gate per round.
\label{fig:circuit}
}
\end{figure}

We are now ready to describe a fault-tolerant implementation of a logical circuit in the Clifford$+T$ basis.
Our protocol consists of an alternating sequence of  rounds
labeled $C$ and $T$, see Fig.~\ref{fig:circuit}.
Each $C$-round is responsible for measuring syndromes of
all face-type stabilizers of the $C$-code, that is,
$X(f)$ and $Z(f)$, where $f$ is one of the vectors $f^i[A]$, $f^i[B]$,
or $\omega[B]+C$. 
Each $T$-round is responsible for measuring syndromes of all
edge-type stabilizers of the $T$-code, that is, stabilizers $Z(l[A]+l[B])$ with $l\in E$.
All measured syndromes are sent to the decoder.
Typically, but not always, a logical Clifford gate  ($T$-gate)
 is applied after each $C$-round ($T$-round).
Whether or not a logical gate is applied 
depends on the outcome of a certain test that we call a syndrome test.
A decoder is  responsible for choosing a recovery operator $R$
which is applied at the end of every pair of $C,T$ rounds  passing the syndrome test.
In the beginning of each round the decoder
performs a code deformation such that the logical qubit is encoded
by the $C$-code ($T$-code) in every $C$-round ($T$-round). 
Although the base code does not explicitly appear in the protocol,
it is used by the decoder as an intermediate step in the code deformation,
see Section~\ref{sec:ML}.
Namely, at the beginning of each $C$-round the gauge group 
changes according to 
\[
\css{\calT}{\dot{\calT}}\to \css{\calC}{\dot{\calT}}  \to \css{\calC}{\calC}.
\]
At the beginning of each $T$-round the gauge group changes in the reverse
direction. 

Let us now describe the syndrome test.
Recall that  the syndromes of $X(f)$ and $Z(f)$ are denoted $\xi(f)$ and $\zeta(f)$
respectively, see Section~\ref{sec:subs}.
Consider some fixed  pair of rounds $C,T$ and let $U$ be the logical Clifford
gate applied in the $C$-round (set $U=I$ if no logical gate have been applied).
Let $\xi(f)$ and $\zeta(f)$ be the face-type syndromes measured in this round.
Measuring the syndrome of a stabilizer $Z(f)$ after application of $U$
is equivalent to measuring the syndrome of $P(f)$ before application of $U$,
where $P\equiv UZU^\dag$. Suppose $P\sim X(a)Z(b)$, where $a,b\in \FF_2$.
Define an updated syndrome
\[
\zeta_U(f)=a\xi(f)+b\zeta(f).
\]
Thus $\zeta_U(f)$ determines the  syndrome of a stabilizer $Z(f)$ that would be observed
in the absence of the logical  gate $U$.
Let $\zeta(l[A]+l[B])$ be the edge-type syndromes measured in the 
 $T$-round. 
We say that the pair of rounds $C,T$ passes a syndrome test if 
\begin{equation}
\label{Stest}
\zeta(l[A]+l[B])+\zeta(l'[A]+l'[B]) + \zeta_U(f^i[A]) + \zeta_U(f^i[B])=0
\end{equation}
for any face $f^i$ and for any pair of edges $l,l'\in E$ such that 
$l+l'=f^i$. In other words, $l$ and $l'$ are the two opposite edges
forming the boundary of $f^i$. 

In the absence of errors the syndrome test is always passed since
the product of edge-type stabilizers
$Z(l[A]+l[B])$, $Z(l'[A]+l'[B])$ and  face-type stabilizers $Z(f^i[A])$, $Z(f^i[B])$ equals the identity.
Our protocol performs the syndrome test after each $T$-round,
see Fig.~\ref{fig:circuit}.
If the syndrome test fails, an additional pair of  rounds $C,T$
is requested and the process continues until the syndrome test is passed.
We note that the syndrome test can fail for at least two reasons. First, any single-qubit $X$-error
on some qubit $j\in AB$  that occurs inside the chosen $T$-round
flips the syndromes $\zeta(l[A]+l[B])$ on all edges $l$ incident to the site $u\in \Lambda$
that contains $j$
without flipping any face-type syndromes (because the latter have been measured {\em before}
this error occurred). This would violate at least one constraint in Eq.~(\ref{Stest}).
Secondly, any single measurement  error
for edge-type stabilizers  $Z(l[A] + l[B])$  in the $T$-round
or any single measurement error for 
 face-type stabilizers  that contribute to the updated syndromes $\zeta_U(f)$
 would violate at least one constraint in Eq.~(\ref{Stest}).
The purpose of the syndrome test is to ensure that 
neither of these possibilities occurs before asking  the decoder to perform 
a recovery operation. 


Combining syndromes measured in any consecutive pair of rounds $C,T$
provides the full syndrome for the $T$-code (in the absence of errors).
Indeed, the syndrome of $X(\calT)$ can be inferred from the
syndromes of $X(\calC)$ and $Z(\calC)$  measured in the $C$-round since $\calT\subseteq \calC$
and $X(\calT)$ commutes  with all operators measured in the $T$-round.
Likewise, the  syndrome of $\vphantom{\hat{\hat{A}}} Z(\dot{\calT})$ can be obtained by combining
the syndromes of $X(\calC)$ and $Z(\calC)$ measured in the $C$-round
and the syndrome of  $Z(\calG)$ measured in the $T$-round,
since $\vphantom{\hat{\hat{A}}} \dot{\calT}=\calC+\calG$.
By spreading the syndrome measurement for the $T$-code over two  rounds 
we were able to keep the number of measurements per qubit in any  single round reasonably small,
which might be important for practical implementation.

A transversal Clifford gate is applied after each $C$-round
provided that the latest syndrome tests was successful.
We performed simulations for a random Clifford$+T$ circuit such that 
each Clifford gate is drawn from the uniform distribution on the Clifford group.

A transversal $T$-gate is applied at the end of each $T$-round that passes the syndrome test.
It is preceded by a Pauli recovery operator $R$ classically controlled by the decoder,
see Section~\ref{sec:ML},
and followed by a twirling  map $\calW_\calT$ that applies a randomly chosen stabilizer $X(f)$ with $f\in \calT$,
\begin{equation}
\label{twirl}
\calW_\calT(\rho)=\frac1{|\calT|} \sum_{f\in \calT} X(f) \rho X(f).
\end{equation}

Each round includes  memory and syndrome measurement errors. 
We model memory errors by the depolarizing noise with some
error rate $p$, that is, each qubit suffers from a Pauli error $X,Y,Z$ with probability $p/3$ each.
Within each round a memory error occurs  before the syndrome measurement, 
see Fig.~\ref{fig:circuit}.
We model a noisy syndrome measurement by an ideal measurement in which the
outcome is flipped with a probability $p$.

At any given time step the protocol can be terminated depending on 
the outcome of two  tests: (1) logical error test and (2) cleanability test.
A logical error test is performed at the end of each round by 
computing the most likely coset of errors consistent with the current syndrome.
The test is passed if the most likely coset contains the actual memory error.
A cleanability test is performed after each recovery operation. 
If $E\sim X(a)Z(b)$ is the residual error left  after the recovery, the test
is passed if $a+\calT$ is a cleanable coset, see Eq.~(\ref{Tspace}) and Definition~\ref{dfn:clean}.
We found that $\calT$ has $996$ cleanable cosets. 
The protocol terminates whenever one of the two tests fails. 
Accordingly, the number of logical gates implemented in the protocol
is a random variable. 
Conditioned on passing the cleanability test, a transversal $T$-gate
is implemented using the method of Section~\ref{sec:Tgate}.
The quantity we are interested in is 
a logical error rate defined as  $p_L=1/g$,
where $g$ is the average number of logical gates implemented before 
the protocol terminates.
Here $g$ includes both Clifford and $T$ gates.

\begin{figure}[h]
\centerline{\includegraphics[height=7cm]{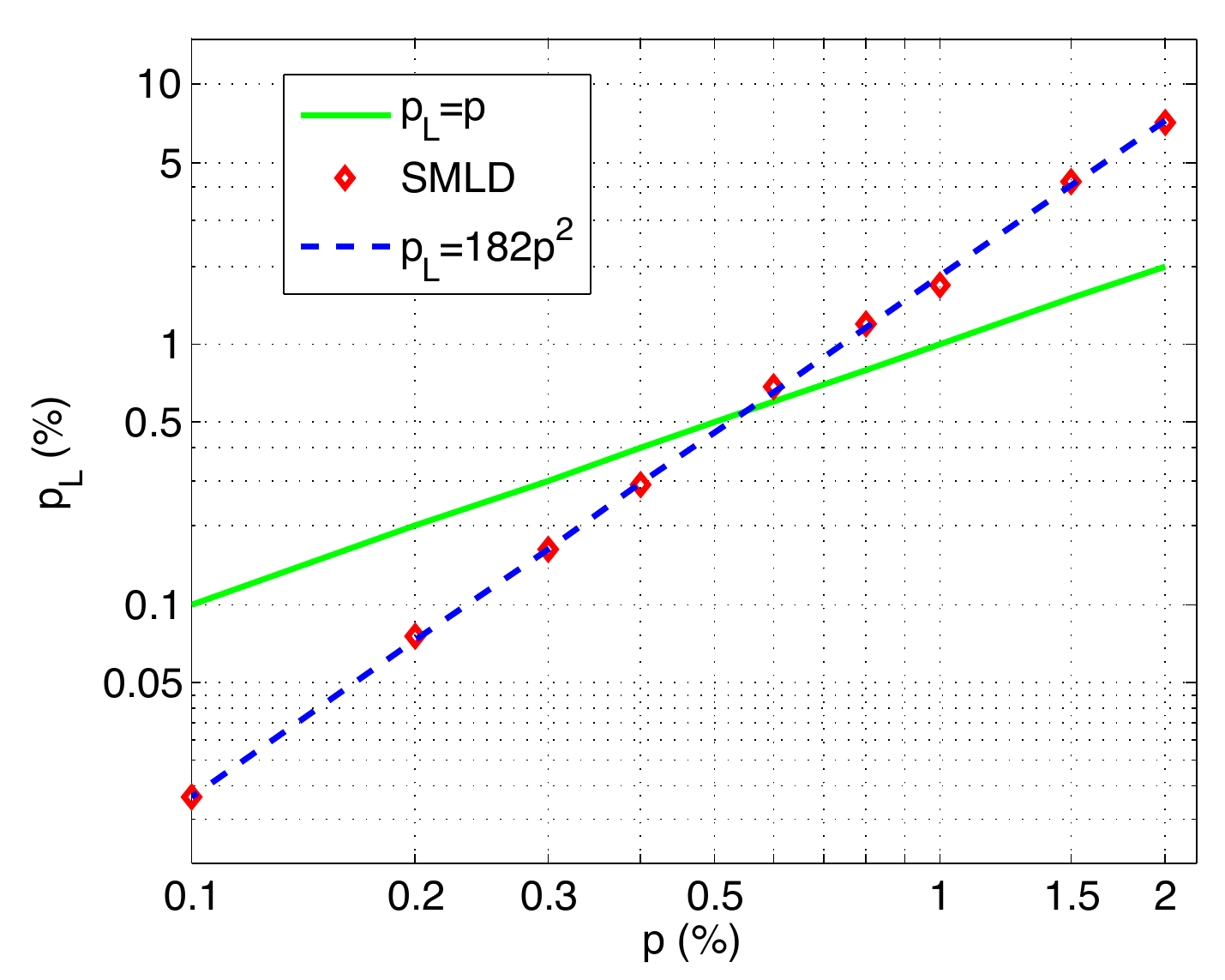}}
\caption{Monte Carlo simulation of logical Clifford$+T$ circuits
with the sparse ML decoder (SMLD).
Here $p$ is the physical error rate 
and $p_L$ is the logical error rate defined as 
$p_L=1/g$, where $g$ is the average number of logical gates 
implemented before the protocol terminates. 
Each average was estimated  using  $400$  Monte Carlo trials.
 \label{fig:CTplot}
}
\end{figure}

The logical error rate $p_L$ was computed numerically  using  a simplified  version
of the ML decoder that we call a sparse ML decoder (SMLD).
It follows the algorithm described in Section~\ref{sec:ML}
with two modifications. First,  SMLD models memory errors by 
a sparse distribution $\pi^s$
approximating the exact distribution $\pi$. 
We have chosen $\pi^s(e)=\pi(e)$ if  $e$ is a Pauli error 
acting non-trivially on at most one qubit 
and $\pi^s(e)=0$ otherwise (here we ignore the normalization). 
Replacing $\pi$ by $\pi^s$ in the update rules of Section~\ref{sec:ML}
one can see that the matrix $P_t$ defined in Eqs.~(\ref{P_t},\ref{PQ}) becomes sparse
and  the update $|\rho_t\rangle \gets P_t |\rho_t\rangle$
can be performed using sparse matrix-vector
multiplication avoiding the Walsh-Hadamard transforms.
The actual memory errors in the Monte Carlo simulation
are drawn from the exact distribution $\pi$.
This approximation is justified since none of the codes used in the protocol 
can correct memory errors of weight larger than one. 
Secondly, SMLD performs a truncation of the likelihood vector 
$\rho_t$   after each round in order to keep $\rho_t$ sufficiently sparse.
The truncation was performed by normalizing $\rho_t$ and setting to zero
all components with $\rho_t(f)<\epsilon$, where  $\epsilon=10^{-6}$
is an empirically chosen  cutoff value.
We observed that for large error rates ($p\approx 1\%$) the two
versions of the decoder achieve the same logical error probability
within statistical fluctuations. On the other hand, SMLD 
provides at least $10$x speedup compared with the exact version
and enables simulation of circuits with more than 10,000 logical gates. 
Our results are presented on Fig.~\ref{fig:CTplot}.
We observed a scaling
$p_L=Cp^2$ with $C\approx 182$. 
Assuming that a physical Clifford$+T$ circuit
has an error probability $p$ per gate,   the logical circuit becomes more reliable than the 
physical one  provided  that  $p_L<p$, that is, $p<p_0=C^{-1} \approx 0.55\%$.
This value can be viewed as an ``error threshold" of the proposed protocol. 
Generating the data shown on Fig.~\ref{fig:CTplot} took approximately one day
on a laptop computer.

\section{Doubled color codes: main properties}
\label{sec:summary}

Our goal for the rest of the paper is to generalize the $15$-qubit codes described in Section~\ref{sec:CT15}
to higher distance codes that can be  embedded into a 2D lattice such that 
all syndrome measurements  required for error correction and gauge fixing are spatially local.
This section highlights main properties of the new codes.
Let $d=2t+1$ be the desired code distance, where $t\ge 1$ is an integer. 
For each $t$ we shall construct a pair of  CSS codes
labeled $C$ and $T$ that
encode one logical qubit into $n_t=2t^3+8t^2+6t-1$ physical qubits.
These codes have transversal Clifford gates and the $T$-gate respectively. 
The codes are defined 
on the 2D honeycomb lattice with two qubits per site such that 
the gauge group of the $C$-code has spatially local generators supported
on faces of the lattice. Most of these generators are analogous to face-type
stabilizers in the $15$-qubit example, see Section~\ref{sec:CT15}.
There are also additional generators of weight two that couple 
pair of qubits located on the same face. The latter have no analogue in the $15$-qubit
example. 
 The $C$-code can be converted to the regular color code on the honeycomb lattice
by discarding a certain subset of qubits. 
Although the $T$-code does not have local gauge generators, it
can be obtained from the $C$-code by a
local gauge fixing. It requires syndrome measurements for weight-four
stabilizers supported on edges of the lattice. 
The latter are analogous to edge-type stabilizers in the $15$-qubit example.
We shall also define a base code that appears as an intermediate step in the 
conversion between $C$ and $T$ codes. 
All three codes have distance $d=2t+1$. 

The  codes will be constructed from a pair  of linear subspaces
$\calC_t,\calT_t\subseteq \FF_2^{n_t}$ that satisfy
\begin{equation}
\label{desired1}
\calT_t\subseteq \calC_t \subseteq \dot{\calC}_t \subseteq \dot{\calT}_t
\end{equation}
and
\begin{equation}
\label{desired2}
d(\calT_t)=2t+1.
\end{equation}
The subspaces $\calC_t$ and $\calT_t$ have a special symmetry 
required for  transversality of logical gates,  namely,
$\calC_t$ is doubly even and $\calT_t$ is triply even with respect to some
subsets of qubits $M^{\pm}_t\subseteq [n_t]$ such that
$|M^+_t|-|M^-_t|=1$.
Furthermore, $\dot{\calC}_t$ and $\dot{\calT}_t$ have spatially local generators (basis vectors)
supported on faces of the lattice.
The subspaces $\calC_1$ and $\calT_1$ coincide with $\calC$ and $\calT$
defined in Section~\ref{sec:CT15}.
Definitions of  the three codes are summarized in Table~\ref{table:CT}.
We shall refer to the family of codes defined in Table~\ref{table:CT} 
as doubled color codes since each of them is constructed from two copies of the regular color code.

\begin{table}[!ht]
\centerline{
\begin{tabular}{r|c|c|c|}
 & Transversal gates  & Stabilizer group & Gauge group \\
\hline
$\vphantom{\hat{\hat{A}}}$
$C$-code & Clifford group &  $\css{\calC_t}{\calC_t}$ & $\css{\dot{\calC}_t}{\dot{\calC}_t} $   \\
\hline
$\vphantom{\hat{\hat{A}}}$
$T$-code & $T$ gate &  $\css{\calT_t}{\dot{\calT}_t}$ & $\css{\calT_t}{\dot{\calT}_t}$  \\
\hline
$\vphantom{\hat{\hat{A}}}$
Base code &   & $\css{\calT_t}{\calC_t}$  & $\css{\dot{\calC}_t}{\dot{\calT}_t}$ \\
\hline
\end{tabular}}
\caption{A family of 2D doubled color codes that achieves universality by the gauge fixing method.
Each code has one logical qubit,  $n_t=2t^3+8t^2+6t-1$ physical qubits,
and  distance $d=2t+1$.
The subspaces $\dot{\calC}_t$ and $\dot{\calT}_t$ have spatially local generators
supported on faces of the lattice. These subspaces are doubly even and triply even
respectively.
}
\label{table:CT}
\end{table}
\noindent

We expect that the new codes can be used 
in a protocol  analogous to the one shown on Fig.~\ref{fig:circuit}
to implement fault-tolerant Clifford$+T$ circuits in the 2D architecture. 
The $C$-round of the protocol would be responsible for measuring
gauge syndromes $\xi(f^\alpha)$ and $\zeta(f^\alpha)$ for some complete set of generators  $f^1,\ldots,f^m\in \dot{\calC}_t$.
A transversal Clifford gate can be applied after each $C$-round. 
The $T$-round would be responsible for measuring 
stabilizer syndromes $\zeta(g^\beta)$ for some set of generators $g^1,\ldots,g^k\in \dot{\calT}_t$
such that  $\dot{\calT}_t=\langle f^1,\ldots,f^m,g^1,\ldots,g^k\rangle$.
We will see that such generators $g^\beta$  can be chosen as double edges of the lattice,
by analogy with the $15$-qubit example. 
A transversal $T$-gate can be applied after each $T$-round. 
In the case of higher distance codes the $C$-round 
could be repeated 
 several times
to ensure that any combination of $t$ syndrome
measurement errors is correctable. 
We expect that the $T$-round requires less (if any) repetitions since
the edge-type stabilizers measured in this round are highly redundant. 
An explicit construction of a fault-tolerant protocol that would suppress the 
error probability per gate 
 from $p$ to $O(p^{t+1})$   is an interesting open problem that we leave for a future work.

\section{Regular color codes}
\label{sec:color}

The starting point for our construction is the standard color code on the hexagonal lattice~\cite{Bombin2006}.
For  a suitable choice of boundary conditions
such code has exactly one logical qubit and corrects $t$ single-qubit errors, where $t\ge 0$
is any given integer (recall that the code distance is $d=2t+1$). 
Examples of the color code lattice for $t\le 3$ are shown on Fig.~\ref{fig:Lambda}.
More formally, consider a  lattice $\Delta_t$ such that the lattice sites 
are  triples of non-negative integers $\jbf=(j_1,j_2,j_3)$ satisfying
$j_1+j_2+j_3=3t$. Clearly, $\Delta_t$ is the regular triangular lattice.
Note that  $j_2-j_1 = j_3-j_2 = j_1-j_3{\pmod 3}$
for any $\jbf\in \Delta_t$. For each $b\in \ZZ_3$ consider a sublattice
\begin{equation}
\label{lattice1}
\Delta_t^b=\{\; \jbf \in \Delta_t\, : \, j_2-j_1=b {\pmod 3}\; \}.
\end{equation}
Define a color code lattice $\Lambda_t$ as
\begin{equation}
\label{lattice2}
\Lambda_t=\Delta_t^0 \cup \Delta_t^2.
\end{equation}
In other words, sites of $\Lambda_t$ are triples of non-negative integers $\jbf=(j_1,j_2,j_3)$
satisfying $j_1+j_2+j_3=3t$ and $j_2-j_1\ne 1{\pmod 3}$.
Sites of the triangular lattice $\Delta_t$ that are not present in $\Lambda_t$ 
become centers of faces of $\Lambda_t$.
More formally, a subset $f\subseteq \Lambda_t$ is called a face iff
there exists $\jbf \in \Delta_t^1$ such that $f$ is the set of 
nearest neighbors of $\jbf$ in the triangular lattice $\Delta_t$.
A direct inspection shows that any face of $\Lambda_t$ 
consists of four or six sites. 
Although the color code lattice is usually equipped with a face $3$-coloring,
we shall not specify the colors since they are irrelevant for what follows.

\begin{figure}[h]
\centerline{\includegraphics[height=4cm]{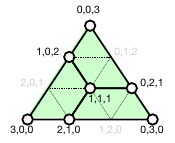}
\raisebox{4mm}{\includegraphics[height=4cm]{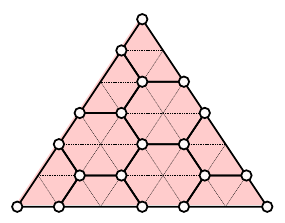}}
\raisebox{3mm}{\includegraphics[height=6cm]{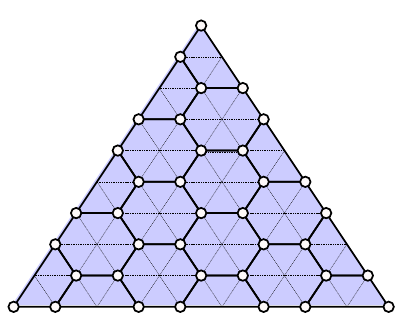}}}
\caption{Color code lattices $\Lambda_1$, $\Lambda_2$, and $\Lambda_3$.
Qubits are located at the lattice sites (empty circles).
 Each face $f$
gives rise to a pair of stabilizer generators $\hat{X}(f)$ and $\hat{Z}(f)$.
The color code defined on the lattice $\Lambda_t$ has one logical qubit
and corrects $t$ single-qubit errors (code distance is $d=2t+1$).
\label{fig:Lambda}
}
\end{figure}

The lattice $\Lambda_t$ contains 
\begin{equation}
m_t=|\Lambda_t|=3t^2+3t+1
\end{equation}
sites and $(m_t-1)/2$ faces. We will use a convention that $\Lambda_0$ is a single site.
Let  $\calS_t \subseteq \FF_2^{m_t}$ be a
subspace spanned by all faces of $\Lambda_t$
 (recall that we identify a face $f$ and a binary vector
whose support is $f$). Note that $\calS_t$ is self-orthogonal,
since any pair of faces overlap on even number of sites.
We shall need the following well-known properties of $\calS_t$,
see Refs.~\cite{Bombin2006,Bombin2015,Kubica2015}.
\begin{fact}
\label{fact:CC}
The subspace $\calS_t$ is doubly-even with respect to the
subsets $\Delta^0_t$ and $\Delta^2_t$.
Furthermore, 
\begin{equation}
\label{CCfact1}
|\Delta^0_t|-|\Delta^2_t|=1.
\end{equation}
The orthogonal subspace $\calS_t^\perp$ is given by
\begin{equation}
\label{CCfact2}
\calS_t^\perp \cap \calE=\calS_t \quad \mbox{and}  \quad \calS_t^\perp \cap \calO=\calS_t + \overline{1},
\end{equation}
Finally,
\begin{equation}
\label{CCfact3}
d(\calS_t)=2t+1.
\end{equation}
\end{fact}
For the sake of completeness, let us prove the first claim.
\begin{proof}
We shall use shorthand notations $\calS\equiv \calS_t$, $m\equiv m_t$, and $\Delta^b\equiv \Delta^b_t$.
Consider any stabilizer $f\in \calS$. Then  $f$ is a linear combination of faces.
One can always choose the numbering of  faces such that 
$f=\sum_{i=1}^k f^i$ for some $k\le s$.
We shall use an identity
\begin{equation}
\label{id1}
\left| \sum_{i=1}^k g^i \right| = \sum_{i=1}^k |g^i | - 2\sum_{1\le i<j\le k}\; |g^i \cap g^j|
{\pmod 4}
\end{equation}
which holds for any vectors $g^1,\ldots,g^k\in \FF_2^m$.
Choosing $g^i=f^i \cap \Delta^b$  one gets
\begin{equation}
\label{Sgate2}
|f\cap  \Delta^0| - |f\cap \Delta^2| =\sum_{i=1}^k \left( |f^i\cap  \Delta^0| - |f^i\cap \Delta^2|\right)
-2\sum_{1\le i<j\le k}\;  \left( |f^i \cap f^j\cap \Delta^0| -  |f^i \cap f^j\cap \Delta^2| \right) 
\end{equation}
modulo four. Consider any site $\jbf \in \Delta^1$. 
The nearest neighbors of $\jbf$ in the triangular lattice belong to either
$\Delta^0$ or $\Delta^2$ as shown on Fig.~\ref{fig:arrows}.
By examining this figure one can easily check 
that any edge of the color code lattice $\Lambda_t$
connects some site of $\Delta^0$ and some site $\Delta^2$. 

\begin{figure}[h]
\centerline{\includegraphics[height=3cm]{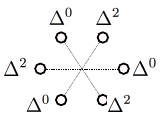}}
\caption{A local neighborhood of a site $\jbf \in \Delta^1$.
This shows that any edge of the color code lattice  $\Lambda_t=\Delta^0\cup \Delta^2$ 
has one endpoint in $\Delta^0$ and the other endpoint in $\Delta^2$.
\label{fig:arrows}
}
\end{figure}

It follows that 
$|f^i\cap  \Delta^0| - |f^i\cap \Delta^2|=3-3=0$
for hexagonal faces $f^i$ and
$|f^i\cap  \Delta^0| - |f^i\cap \Delta^2|=2-2=0$
for square faces $f^i$. Likewise, consider any pair of faces $f^i,f^j$ such that 
$f^i \cap f^j\ne \emptyset$. Then $f^i \cap f^j$ is an edge of the color code lattice
$\Lambda_t$, so that 
\[
|f^i \cap f^j\cap \Delta^0| -  |f^i \cap f^j\cap \Delta^2|=1-1=0.
\]
This proves that all terms in Eq.~(\ref{Sgate2}) are zero,
that is, $\calS$ is triply-even as promised. 

\end{proof}

Finally, let $\omega^1_t,\omega^2_t$ and $\omega^3_t$ be the vectors
supported on  the right, on the left, and on the bottom sides of the  triangle
that forms the boundary of $\Lambda_t$.
In other words, 
\begin{equation}
\label{omega}
\omega^i_t=\{\jbf\in \Lambda_t\, : \, j_i=0\}, \quad \quad i=1,2,3.
\end{equation}
Note that $\omega^i_t$ are logical operators, that is, $\omega^i_t\in \calS_t^\perp \cap \calO$.
Furthermore, $|\omega^i_t|=2t+1\equiv d$,
that is,  $\omega^i_t$ are minimum weight
logical operators.

\section{Doubling transformation}
\label{sec:doubling}

Let us now describe a general construction of triply even subspaces 
inspired by  Ref.~\cite{Betsum2010}.
 Consider a pair of integers  $m,n$ and a  pair of subspaces
\[
\calS\subseteq \FF_2^m \quad \mbox{and} \quad \calT\subseteq \FF_2^n.
\]
Let $k=2m+n$. Partition the set of integers $[k]$ into three consecutive blocks,
$[k]=ABC$,  
such that $|A|=|B|=m$ and $|C|=n$.
Define a subspace 
$\calU\subseteq \FF_2^k$ such that 
\begin{equation}
\label{Sout}
\calU=\langle f[A]+f[B] \, : \, f\in \calS\rangle + 
\calT[C]
+ \langle BC \rangle.
\end{equation}
Here we used a shorthand notation $BC\equiv \overline{1}[BC]$.
This definition can be rephrased  in terms of the generating matrices of $\calS$ and $\calT$.
Recall that $S$ is a generating matrix of a linear subspace $\calS$
if $\calS$ is spanned by rows of $S$. 
Suppose  $\calS$ and $\calT$ have generating matrices $S$ and $T$
respectively.  Then $\calU$ has  a generating matrix
\begin{equation}
\label{S1'}
U=\left[
\begin{array}{ccc}
S & S & \\
 & & T \\
 & \overline{1} & \overline{1} \\
\end{array} \right]
\end{equation}
where the three groups of columns correspond to $A,B$, and $C$
respectively.
Here we only show the non-zero entries of $U$.
The mapping from $\calS$ and $\calT$ to $\calU$ will be referred to as
a doubling map since it involves taking two copies of $\calS$.
We shall use an informal notation $\calU=2\calS+\calT$ to indicate
that $\calU$ is obtained from $\calS$ and $\calT$ via the
doubling map as described above.

Given a subset $M\subseteq [m]$ and a subset $N\subseteq [n]$
let $MMN\subseteq ABC$ be a subset obtained by choosing
the subset $M$ in the blocks $A,B$ and choosing the subset $N$ in the block $C$.
\begin{lemma}
\label{lemma:doubling}
Assume $\calS$ is doubly even with respect to some subsets
$M^\pm \subseteq [m]$ and $\calT$ is triply even with respect to
some subsets  $N^\pm \subseteq [n]$ such that 
\begin{equation}
\label{mod8condition}
|M^+|-|M^-| + |N^-|-|N^+|=0{\pmod 8}.
\end{equation}
Then $\calU=2\calS+\calT$ is triply even with respect to subsets $K^\pm = M^\pm M^\pm N^\mp$.
\end{lemma}
\begin{proof}
Consider first a vector
\[
h=f[A]+f[B]+g[C]\in \calU,
\]
where $f\in \calS$ and $g\in \calT$. 
By assumption, 
\begin{equation}
\label{h2}
|f\cap M^+ | -  |f\cap M^- |=0{\pmod 4} \quad \mbox{and} \quad |g\cap N^+|-|g\cap N^-|=0{\pmod 8}.
\end{equation}
The identity $|h\cap K^\pm |=2|f\cap M^\pm | + |g\cap N^\mp |$ then implies 
\[
|h\cap K^+|-|h\cap K^-| =0{\pmod 8}.
\]
Consider now a vector 
\[
h'=h+ BC\in \calU.
\]
Then 
\[
|h'\cap K^\pm | =|f\cap M^\pm |  +  (|M^\pm | -|f\cap M^\pm | ) + (|N^\mp| - |g\cap N^\mp |)=
|M^\pm |+ |N^\mp| - |g\cap N^\mp |.
\]
Taking into account Eqs.~(\ref{mod8condition},\ref{h2}) one arrives at
\[
|h'\cap K^+|-|h'\cap K^-| =0{\pmod 8}.
\]
Since any vector of $\calU$ can be written as $h$ or $h'$,
the lemma is proved.
\end{proof}

Next let us compute the orthogonal subspace $\calU^\perp$.
We specialize to the case when 
 $m,n$ are odd, whereas all vectors
in $\calS$ and $\calT$ have even weight.
This will be the case for applications considered below,
where $\calS$ and $\calT$ define stabilizer groups of CSS-type quantum codes.
\begin{lemma}
\label{lemma:ort}
Suppose  $n$ and $m$ are odd. Suppose 
$\calS\subseteq \calE$ and $\calT\subseteq \calE$.
Then 
\begin{equation}
\label{Uperp}
\calU^\perp = \langle f[A] + f[B]  \, : \, f\in \calE\rangle 
+ \calS^\perp[A] + \dot{\calT}[C]+ 
 \langle BC\rangle.
\end{equation}
Furthermore,
\begin{equation}
\label{UperpEven}
\dot{\calU}= 
 \langle f[A] + f[B]  \, : \, f\in \calE\rangle  + 
\dot{\calS}[B] + \dot{\calT}[C]+ 
 \langle BC\rangle.
\end{equation}
and
\begin{equation}
\label{Udistance}
d(\calU)=\min{ \{ d(\calS), d(\calT)+2\}}.
\end{equation}
\end{lemma}
\begin{proof}
Consider an arbitrary vector $h\in \FF_2^k$.
By definition,   the inclusion $h\in \calU^\perp$ is equivalent to
\begin{equation}
\label{ort_cond}
h_A+h_B \in \calS^\perp, \quad h_C\in \calT^\perp, \quad \mbox{and} \quad
h_B\oplus h_C\in \calE.
\end{equation}
A direct inspection shows that Eq.~(\ref{ort_cond}) holds 
if $h$ belongs to each individual term in Eq.~(\ref{Uperp})
which proves the inclusion $\supseteq$ in Eq.~(\ref{Uperp}).
Conversely, suppose $h\in \calU^\perp$.
We have to prove that $h$ is contained in the sum of the four terms in Eq.~(\ref{Uperp}).
Set $f=h_B$ if $h_B\in \calE$ and $f=h_B+\overline{1}$ if $h_B\in \calO$.
In both cases $f\in \calE$. Replacing $h$ by $h+f[A]+f[B]$
we can make $h_B=\overline{0}$ or $h_B=\overline{1}$.
Since $\overline{1}\in \calS^\perp$, the first condition in Eq.~(\ref{ort_cond})
implies $h_A\in \calS^\perp$. If $h_B=\overline{1}$,
replace $h$ by $h+\overline{1}[B]+  \overline{1}[C]$.
This makes $h_B=\overline{0}$ and does not change the second condition in Eq.~(\ref{ort_cond})
since $\overline{1}\in \calT^\perp$. The last condition in Eq.~(\ref{ort_cond})
then implies $h_C\in \calE$. We conclude that $h=g[A] + g'[C]$
for some $g\in \calS^\perp$ and $g'\in \dot{\calT}$.
This proves Eq.~(\ref{Uperp}). 

To prove Eq.~(\ref{UperpEven}) we note 
that odd-weight vectors in $\calU^\perp$ can only originate from
the term  $\calS^\perp[A]$.
Restriction to the even subspace replaces this term
by $\dot{\calS}[A]$.
Since we already know that $\dot{\calU}$ contains all
vectors $f[A]+f[B]$ with  $f\in \calE$, we can move
$\dot{\calS}$ from $A$ to $B$ without changing $\dot{\calU}$.
This proves Eq.~(\ref{UperpEven}). 

It remains to prove Eq.~(\ref{Udistance}). 
Choose any vectors $f^*\in \calS^\perp \cap \calO$ and $g^*\in \calT^\perp \cap \calO$ such that $d(\calS)=|f^*|$ and $d(\calT)=|g^*|$. 
Choose any $i\in [m]$ and consider vectors $x=e^i[A]+ e^i[B] + g^*[C]$
and $y=f^*[A]$. 
A direct inspection shows that $x,y\in \calU^\perp \cap \calO$.
Therefore
$d(\calU)\le |x|=2+|g^*|=d(\calT)+2$
and 
$d(\calU)\le |y|=|f^*|= d(\calS)$. This proves the  inequality $\le$ in Eq.~(\ref{Udistance}). 
Let us prove the reverse inequality. Consider any vector $h\in \calU^\perp \cap \calO$.
It must satisfy Eq.~(\ref{ort_cond}).
Since $h\in \calO$,  the  condition $h_B\oplus h_C\in \calE$   implies  $h_A\in \calO$.
Consider two cases. Case~1: $h_C\in \calO$. Then the second   condition in Eq.~(\ref{ort_cond})
implies $h_C\in \calT^\perp \cap \calO$, that is, $|h_C|\ge d(\calT)$. 
Furthermore, since $h,h_C,h_A\in \calO$ we infer that $h_B\in \calO$ which implies
$|h|=|h_A|+|h_B|+|h_C|\ge 2+d(\calT)$. Case~2: $h_C\in \calE$.  Then 
the  condition $h_B\oplus h_C\in \calE$   implies
 $h_B\in \calE$, that is, $h_A+h_B\in \calO$.
 The first condition in Eq.~(\ref{ort_cond}) implies
$h_A+h_B\in \calS^\perp\cap \calO$, that is, $|h_A+h_B|\ge d(\calS)$. By triangle inequality, 
$|h_A|+|h_B|\ge d(\calS)$ and thus $|h|\ge d(\calS)$. 
This proves  the  inequality $\ge$ in Eq.~(\ref{Udistance}).
\end{proof}
The above lemma has the following obvious corollary.
\begin{corol}
\label{corol:ort}
Suppose  $n$ and $m$ are odd,
$\calS\subseteq \calE$, and $\calT$ is self-orthogonal.
Then $\calU$ is self-orthogonal. 
\end{corol}
\begin{proof}
Self-orthogonality of $\calT$ implies $\calT\subseteq \dot{\calT}$.
Thus the first, the second, and the third terms of Eq.~(\ref{Sout})
are contained in the first, the third, and the fourth terms of Eq.~(\ref{Uperp})
respectively.
\end{proof}

\section{Doubled color codes: construction}
\label{sec:color2}

Let  $\calS_t$ be the subspace spanned by faces of the color code lattice
$\Lambda_t$ constructed in Section~\ref{sec:color}.
Recall that $\calS_t\subseteq \FF_2^{m_t}$ where $m_t\equiv |\Lambda_t|=3t^2+3t+1$. 
A doubled color code with distance $d=2t+1$ will require $n_t$ physical qubits,
where  $n_0=1$ and
\begin{equation}
\label{induction0}
n_t=2m_t+n_{t-1} \quad \mbox{for $t\ge1$}.
\end{equation}
Solving the recurrence relation gives
\begin{equation}
\label{n_t}
n_t=2t^3 + 6t^2 + 6t + 1.
\end{equation}
For instance, $n_1=15$, $n_2=53$, and $n_3=127$.
Define a  family  of   subspaces
$\calT_t\subseteq \FF_2^{n_t}$ such that
$\calT_0=\langle0\rangle \subseteq \FF_2$ and 
\begin{equation}
\label{induction1}
\calT_t=2\calS_t + \calT_{t-1}, 
\end{equation}
In other words, $\calT_t$ is obtained by applying the doubling map of Section~\ref{sec:doubling}
with $\calS=\calS_t$ and $\calT=\calT_{t-1}$. 
To describe this more explicitly, partition the set of integers $[n_t]$ into $2t+1$
consecutive blocks as
\begin{equation}
\label{AtBt}
[n_t]=A_tB_t\ldots A_2B_2 A_1B_1A_0 \quad \mbox{where} \quad
A_r\cong B_r\cong \Lambda_r.
\end{equation}
In other words, $A_r$ and $B_r$ represent two copies of the color code lattice $\Lambda_r$.
Then 
\begin{equation}
\label{Tgens}
\calT_t=\sum_{r=1}^t \langle f[A_r]+f[B_r]\, : \, f\in \calS_r\rangle
+ \sum_{r=1}^t \langle  B_r A_{r-1}\rangle.
\end{equation}
We can also describe $\calT_t$ by its generating matrix.
Suppose $S_t$ is a generating matrix of $\calS_t$
such that rows of $S_t$ correspond to faces of the color code lattice $\Lambda_t$.
Then $\calT_t$ has a generating matrix 

\vspace{3mm}
\centerline{
$T_t=\; $\begin{tabular}{|cc|cc|cc|cc|cc|c|}
 $A_t$ & $B_t$ & $A_{t-1}$ & $B_{t-1}$ &  &  & $A_2$ & $B_2$ & $A_1$ & $B_1$ & $A_0$ \\
\hline
$S_t$ & $S_t$ &&&&&&&&& \\
&& $S_{t-1}$  & $S_{t-1}$ &&&&&&& \\
&&&& $\cdots$ & $\cdots$ &&&&& \\
&&&&&& $S_2$ &  $S_2$ &&& \\
&&&&&&&& $S_1$  & $S_1$ & \\
 & $\overline{1}$ & $\overline{1}$ & && & && && \\
&& & $\overline{1}$ & $\overline{1}$ & && & &&  \\
&&&& $\cdots$ & $\cdots$ &&&&& \\
&&&& &  $\overline{1}$ & $\overline{1}$ & && & \\
&&&&&& &  $\overline{1}$ & $\overline{1}$ & & \\
&&&&&&&&  & $\overline{1}$ & 1 \\
\hline
\end{tabular}    
}
\vspace{3mm}

\noindent
Here the first line indicates which block of qubits contains a given group of columns.
We shall use the subspace $\calT_t$ to construct the $T$-code
as defined in Table~\ref{table:CT}.
Let us prove that  $\calT_t$ has the properties stated in Section~\ref{sec:summary}.
\begin{lemma}
\label{lemma:Npm}
The subspace $\calT_t$ is triply even with respect to some
subsets $N^\pm_t\subseteq [n_t]$ satisfying 
\begin{equation}
\label{induction3}
|N^+_t|-|N^-_t|=1.
\end{equation}
\end{lemma}
\begin{proof}
We shall use  induction in $t$.
The base of induction, $t=0$, corresponds to $n_0=1$
and $\calT_0=\langle 0\rangle\subseteq \FF_2$. 
Clearly, $\calT_0$ is triply even with respect to subsets 
\begin{equation}
\label{induction5}
N^+_0=\{1\}, \quad N^-_0=\emptyset
\end{equation}
which obey Eq.~(\ref{induction3}).
Consider now an arbitrary $t$. 
We already know that $\calS_t$ is doubly even with respect to the subsets $\Delta^0_t$
and $\Delta^2_t$ such that 
\begin{equation}
\label{0-2'}
|\Delta^0_t|-|\Delta^2_t|=1,
\end{equation}
see Fact~\ref{fact:CC}.
Define
\begin{equation}
\label{AB+-}
A^+_t=B^+_t=\Delta^0_t \quad \mbox{and} \quad A^-_t=B^-_t=\Delta^2_t.
\end{equation}
Here we consider $A^\pm_t$ and $B^\pm_t$ as subsets of $A_t$
and $B_t$ respectively. 
Choose
\begin{equation}
\label{induction4}
N^\pm_t = A^\pm_t B^\pm_t N^\mp_{t-1}, \quad t\ge 1. 
\end{equation}
Combining   Eqs.~(\ref{0-2'},\ref{AB+-},\ref{induction4}) 
and assuming that $|N^+_{t-1}|-|N^-_{t-1}|=1$
one gets
\begin{equation}
\label{induction6}
|N^+_t|-|N^-_t| = 2(|\Delta^2_t| - |\Delta^0_t|) + |N^-_{t-1}|-|N^+_{t-1}|
=1.
\end{equation}
This proves Eq.~(\ref{induction3}) for all $t\ge 0$.
Furthermore, Eqs.~(\ref{induction3},\ref{0-2'},\ref{AB+-}) imply
\[
|A^+_t| - |A^-_t|-|N^+_t|+|N^-_t|=0
\]
for all $t\ge 0$. This shows that  condition Eq.~(\ref{mod8condition}) of Lemma~\ref{lemma:doubling}
is satisfied for $M^\pm = A^\pm_t=B^\pm_t$.
 The lemma implies that  $\calT_t$  is triply even with respect to the
subsets $N^{\pm}_t$ for all $t\ge 0$.
\end{proof}

Let us use induction in $t$ to show that $\calT_t$ has distance
\begin{equation}
\label{d_t}
d(\calT_t)=2t+1.
\end{equation}
Indeed, $d(\calT_0)=1$ since $\calT_0^\perp=\langle0\rangle^\perp=\FF_2$ and  the only odd-weight vector in $\FF_2$ is $1$.
Furthermore, the color code on the lattice $\Lambda_t$ 
has distance $2t+1$, that is, $d(\calS_t)=2t+1$, see Fact~\ref{fact:CC}.
Assuming that $d(\calT_{t-1})=2t-1$ and using
Eq.~(\ref{Udistance}) of Lemma~\ref{lemma:ort} we infer that 
$d(\calT_t)=\min{\{ d(\calS_t),2+d(\calT_{t-1})\}}=2t+1$ which proves Eq.~(\ref{d_t})
for all $t\ge 0$.

To construct the $T$-code we shall also need a subspace $\dot{\calT}_t$, see Table~\ref{table:CT}.
It will be convenient to rewrite the partition in Eq.~(\ref{AtBt})
as $[n_t]=A_tB_tC_t$, where 
\[
C_t=A_{t-1}B_{t-1} \ldots A_1B_1A_0.
\]
Applying Eq.~(\ref{UperpEven}) of Lemma~\ref{lemma:ort} 
and taking into account that $\dot{\calS}_t=\calS_t$, see Fact~\ref{fact:CC},
one gets
$\dot{\calT}_0=0$ and 
\begin{equation}
\label{Gt}
\dot{\calT}_t=\langle f[A_t] + f[B_t]\, : \, f\in \calE\rangle
+ \calS_t[B_t] + \dot{\calT}_{t-1}[C_t] + \langle B_t A_{t-1} \rangle
\end{equation}
for $t\ge 1$. 
 Let $\omega_t\in \calS_t^\perp \cap \calO$
be some fixed minimum weight  logical operator of the regular color code
such that $|\omega_t|=2t+1$.
Later on we shall choose $\omega_t$ as defined in Eq.~(\ref{omega}).
Then $\overline{1}=\omega_t+g$ for some $g\in \calS_t$, see
Eq.~(\ref{CCfact2}) of Fact~\ref{fact:CC}.
This implies $B_t A_{t-1}= \omega_t[B_t] + g[B_t]+A_{t-1}$.
Since $g[B_t]$ is contained in the second term in Eq.~(\ref{Gt}),
we can replace the last term by  $\langle \omega_t[B_t] + A_{t-1}\rangle$.
Likewise, 
\[
A_{t-1}=\omega_{t-1}[A_{t-1}] + f[C_t]
\]
for some $f\in \dot{\calT}_{t-1}$. Here we noted that
 both $\overline{1}[A_{t-1}]$ and
$\omega_{t-1}[A_{t-1}]$ are contained in $\calT_{t-1}^\perp \cap \calO$,
so that the sum of them is contained in $\calT_{t-1}^\perp \cap \calE\equiv  \dot{\calT}_{t-1}$.
 Since $f[C_t]$ is contained
in the third term in Eq.~(\ref{Gt}), we can rewrite the last term 
as $\langle \omega_t[B_t] + \omega_{t-1}[A_{t-1}]\rangle$. Therefore
\begin{equation}
\label{Gt2}
\dot{\calT}_t=\langle f[A_t] + f[B_t]\, : \, f\in \calE\rangle
+ \calS_t[B_t] + \dot{\calT}_{t-1}[C_t] + \langle \omega_t[B_t] + \omega_{t-1}[A_{t-1}]\rangle
\end{equation}
for $t\ge 1$. Note that here we can add a term $\calS_t[A_t]$
since $f[A_t]+f[B_t]$ is contained in the first term in Eq.~(\ref{Gt2}) for 
any $f\in \calS_t$. We can also describe $\dot{\calT}_t$ by its generating matrix.
Suppose $E_t$ is a generating matrix of the even subspace $\calE^{m_t}$
such that rows of $E_t$ correspond to edges of the color code lattice $\Lambda_t$,
that is, each row of $E_t$ has a form $e^u+e^v$ for some edge $(u,v)$ of $\Lambda_t$.
Then $\dot{\calT}_t$ has a generating matrix
\vspace{3mm}
\centerline{
$\dot{T}_t=\; $\begin{tabular}{|cc|cc|cc|cc|cc|c|}
 $A_t$ & $B_t$ & $A_{t-1}$ & $B_{t-1}$ &  &  & $A_2$ & $B_2$ & $A_1$ & $B_1$ & $A_0$ \\
\hline
$S_t$ & &&&&&&&&& \\ 
& $S_t$ &&&&&&&&& \\ 
&& $S_{t-1}$  & &&&&&&&\\
&& & $S_{t-1}$ &&&&&&&\\
&&&& $\cdots$ & $\cdots$ &&&&& \\
&&&&&& $S_2$   &  &&& \\
&&&&&&  &  $S_2$ &&& \\
&&&&&&&&  $S_1$  &  & \\
&&&&&&&&   & $S_1$ & \\
$E_t$ & $E_t$ &&&&&&&&& \\
&& $E_{t-1}$  & $E_{t-1}$ &&&&&&& \\
&&&& $\cdots$ & $\cdots$ &&&&& \\
&&&&&& $E_2$ &  $E_2$ &&& \\
&&&&&&&& $E_1$  & $E_1$ & \\
 & $\omega_t$ & $\omega_{t-1}$ & && & && && \\
&& & $\omega_{t-1}$ & $\omega_{t-2}$ & && & &&  \\
&&&& $\cdots$ & $\cdots$ &&&&& \\
&&&& &  $\omega_3$ & $\omega_2$ & && & \\
&&&&&& &  $\omega_2$ & $\omega_1$ & & \\
&&&&&&&&  & $\omega_1$ & 1 \\
\hline
\end{tabular}    
}
\vspace{3mm}

\noindent
We shall refer to the first two groups of rows as {\em face-type} generators
and {\em edge-type} generators. 

To construct the $C$-code we shall need a subspace $\calC_t\subseteq \FF_2^{n_t}$
defined as 
\begin{equation}
\label{Cgens}
\calC_t=\sum_{r=1}^t \langle f[A_r]+g[B_r]\, : \, f,g\in \calS_r\rangle
+ \sum_{r=1}^t \langle  B_r A_{r-1}\rangle.
\end{equation}
By comparing Eqs.~(\ref{Tgens},\ref{Gt},\ref{Cgens}) one can see that 
\begin{equation}
\label{Cproperty3}
\calT_t\subseteq \calC_t = \dot{\calC}_t \subseteq \dot{\calT}_t.
\end{equation}
A generating matrix of $\calC_t$ can be obtained from $\dot{T}_t$
by removing all edge-type generators. 
A direct inspection shows that 
generators of $\calC_t$ can be partitioned into mutually disjoint subsets
supported on regions 
\[
M_t= A_t, \quad M_{t-1}=B_tA_{t-1}, \quad M_{t-2}=B_{t-1}  A_{t-2},
\quad \ldots \quad
, \quad M_0=B_1 A_0.
\]
Thus we can analyze properties of $\calC_t$ on each region $M_r$ separately.
Generators of $\calC_t$ supported on $M_t$ have
a form $f[A_t]$ with $f\in \calS_t$. 
These generators describe the regular color code $\calS_t$.
Generators supported on a region $M_r$  describe the two-qubit EPR state
$|0,0\rangle+|1,1\rangle$ shared between $B_{r+1}$ and $A_r$
such that the two qubits are encoded by the regular color code
$\calS_{r+1}$ and $\calS_r$.
As a consequence we obtain 
\begin{corol}
\label{corol:Ct}
The subspace $\calC_t$ is doubly even with respect to the
subsets  
$\Delta^{0,2}_t\subseteq A_t$, see Fact~\ref{fact:CC}.
Furthermore, $d(\calC_t)=2t+1$.
\end{corol}
At this point we have proved all properties of the subspaces
$\calC_t,\calT_t$ stated in Section~\ref{sec:summary} except for the spatial locality.
  Let $\Lambda$ be the honeycomb   lattice 
with two qubits per site. 
We shall allocate a triangular-shaped region of $\Lambda$ 
isomorphic to the color code lattice $\Lambda_t$ 
 to accommodate the blocks of qubits $A_t$ and $B_t$. 
 The two blocks can  share the same set of sites
since each site contains two qubits. 
By a slight abuse of terminology, we shall identify $\Lambda_t$
and the region of $\Lambda$  accommodating $A_t$ and $B_t$.
The regions $\Lambda_t$ and $\Lambda_{t-1}$  are placed next to each other
as shown on Fig.~\ref{fig:snake}.

\begin{figure}[h]
\centerline{\includegraphics[height=8cm]{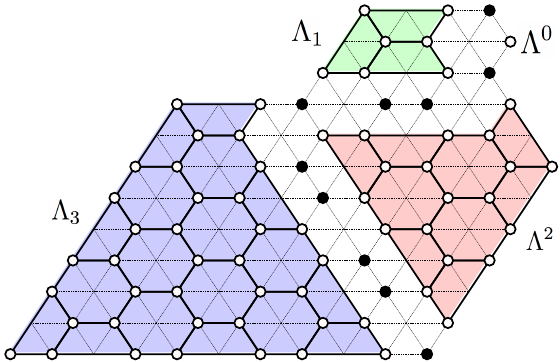}}
\caption{Embedding of the doubled color code $\calT_3$ into the hexagonal lattice $\Lambda$
with two qubits per site.
The lattice is divided into four disjoint regions $\Lambda_0,\Lambda_1,\Lambda_2,\Lambda_3$
where $\Lambda_t$ accommodates two copies of the color code $\calS_t$
labeled $A_t$ and $B_t$ (not shown).
Recall that $\Lambda_0$ is a single site.
To enable a more regular embedding we twisted one corner of each lattice $\Lambda_t$
compared with Fig.~\ref{fig:Lambda}. Sites represented by solid circles can be ignored at this point. 
To obtain the next code $\calT_4$ one should attach $\Lambda_4$ to the bottom
side  of $\Lambda_3$. To obtain $\calT_5$ one should attach $\Lambda_5$ to the left side
of $\Lambda_4$. The process continues in the alternating fashion. 
 Gauge generators live on edges and faces of the lattice.  There are also
non-local gauge generators $\omega_{t,t-1}$ of weight $4t$ connecting 
$B$-qubits on the 
boundary of $\Lambda_t$ and $A$-qubits on the boundary of $\Lambda_{t-1}$ (not shown).
\label{fig:snake}
}
\end{figure}

Consider the generating matrix $\dot{T}_t$ defined above.
We shall say that a vector $f\in \FF_2^{n_t}$ is {\em spatially local}
if its support is contained in a single face of the lattice. 
By definition, face-type and edge-type generators of 
$\dot{T}_t$ are  spatially local. 
Consider some row of $\dot{T}_t$ in the bottom group. It has a form 
\[
\omega_{r,r-1}\equiv \omega_r[B_r] + \omega_{r-1}[A_{r-1}], \quad \quad r=1,\ldots,t.
\]
The above arguments show that we are free to choose different logical operators
$\omega_r$ in different generators $\omega_{r,r-1}$.
Let us use this freedom to choose 
\begin{equation}
\label{omega1}
\omega_{r,r-1}\equiv \omega_r^i[B_r] + \omega_{r-1}^j[A_{r-1}]
\end{equation}
where $\omega_r^i$ is the minimum weight logical operator
supported on the $i$-th boundary of the lattice $\Lambda_r$, see Eq.~(\ref{omega}).
Furthermore, we can choose $i$ and $j$ such that 
the $i$-th boundary of $\Lambda_r$ is located next to the $j$-th boundary of $\Lambda_{r-1}$,
see Fig.~\ref{fig:snake}. Then the generator $\omega_{r,r-1}$ has a shape 
of a  loop that  encloses the free space separating the regions $\Lambda_r$
and $\Lambda_{r-1}$. Since $|\omega_r^i|=2r+1$, we get
$|\omega_{r,r-1}|=2r+1+2r-1=4r$.
We shall explain how to reduce the weight of the generators $\omega_{r,r-1}$ and make then spatially local
in the next section.

The above discussion also shows that a code deformation transforming  the $C$-code
to the $T$-code requires only spatially local syndrome measurements.
Indeed,
suppose $\rho_L$ is some logical state of the $C$-code, see Table~\ref{table:CT}.
One can  first apply a reverse gauge fixing that extends  the gauge 
group of $\rho_L$ from $\css{\calC_t}{\calC_t}$ to $\css{\calC_t}{\dot{\calT}_t}$. This can be achieved by 
applying a random element of the group $Z(\dot{\calT}_t)$.
This also   restricts the stabilizer group 
of $\rho_L$ from $\css{\calC_t}{\calC_t}$ to $\css{{\calT}_t}{\calC_t}$.
We note that 
\[
\dot{\calT}_t=\calC_t + \sum_{l=(u,v)} \langle l[A]+l[B]\rangle,
\]
where the sum runs over all edges of  the sub-lattices $\Lambda_1,\ldots,\Lambda_r$.
Thus a gauge fixing that extends the stabilizer group of $\rho_L$
from $\css{{\calT}_t}{\calC_t}$
to 
$\css{\calT_t}{\dot{\calT}_t}$ can be realized by measuring syndromes
of the edge-type stabilizers $Z(l[A]+l[B])$ and applying a suitable recovery operator.

\section{Weight reduction}
\label{sec:gadget}

In this section we transform  the $C$-code   into a local form.
This requires two steps. 
First, we show how to represent each non-local gauge generator
$\omega_{r,r-1}$, see Eq.~(\ref{omega1}), as a sum
of spatially local generators of weight at most six and a single non-local generator of weight two. 
Secondly, we  show how to represent each of the remaining non-local generators
as a sum of spatially local  generators of weight two.
Each of these steps extends the code by adding several ancillary qubits and 
gauge generators. We add the same ancillary qubits and  gauge generators
to both $C$ and $T$ codes to preserve the local mapping between them. 
Accordingly, we have to prove that the extended versions of the
$C$ and $T$ codes have the same distance and the same
transversality properties as the original codes.

We begin by setting up some notations. Below we consider some fixed
value of $t$ and $r=1,\ldots,t$. 
Let us denote the sites of $\Lambda_r$ lying on the boundary facing $\Lambda_{r-1}$
as $u^1,u^2,\ldots,u^{2r+1}$.
The ordering is chosen such that $u^{2r+1}$ is the ``twisted" corner of $\Lambda_r$,
see Fig.~\ref{fig:gadget}.
Let us denote the sites of $\Lambda_{r-1}$ lying on the 
boundary facing $\Lambda_r$ as $v^1,v^2,\ldots,v^{2r-1}$.
The ordering is chosen such that $u^i$ and $v^i$ are next-to-nearest neighbors,  see Fig.~\ref{fig:gadget}.
Using these notations one can rewrite Eq.~(\ref{omega1}) as
\begin{equation}
\label{bad_guy}
\omega_{r,r-1}= \sum_{i=1}^{2r+1} u^i[B_r] + \sum_{i=1}^{2r-1} v^i[A_{r-1}].
\end{equation}
Consider now sites of $\Lambda$ lying in the free space separating
$\Lambda_r$ and $\Lambda_{r-1}$. These sites are indicated by filled circles
on Figs.~\ref{fig:snake},\ref{fig:gadget}.
Denote these sites as  $w^1,w^2,\ldots,w^{2r}$, see Fig.~\ref{fig:gadget},
and let
\[
D_r=\{w^1,\ldots,w^{2r}\}.
\]
The ordering is chosen such that $w^i$ is a nearest neighbor of $u^i$ and $v^i$
for all $1\le i\le 2r-1$. Furthermore, $w^{2i}$ is a nearest neighbor of $w^{2i+1}$.

\begin{figure}[h]
\centerline{\includegraphics[height=4cm]{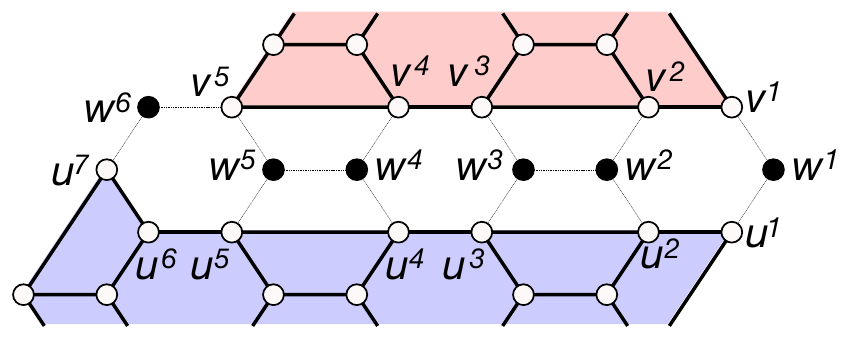}}
\caption{Sites $u^i$ on the boundary of $\Lambda_3$ (blue),
sites $v^i$ on the boundary of $\Lambda_2$ (red),
and sites $w^i$ in the region $D_3$ (filled circles).  
The lattice is rotated $60^\circ$ counter-clockwise
compared with Fig.~\ref{fig:snake}. 
\label{fig:gadget}
}
\end{figure}

Let us add $2r$ ancillary qubits such that each site of $D_r$ contains one qubit.
The total number of physical qubits in the lattice $\Lambda$ becomes
\begin{equation}
\label{N_t}
N_t=n_t+\sum_{r=1}^t 2r=2t^3+7t^2+7t+1,
\end{equation}
see Eq.~(\ref{n_t}). Accordingly, the partition Eq.~(\ref{AtBt}) becomes
\begin{equation}
\label{ABD}
[N_t]=A_tB_tD_t \ldots A_2B_2D_2A_1B_1D_1A_0.
\end{equation}
We shall use  terms $A$-qubit, $B$-qubit, and $D$-qubit
to indicate which of the blocks in Eq.~(\ref{ABD}) contains a given qubit.
Note that the site $w^{2r}$ of $D_r$ and the site $w^1$ of $D_{r-1}$ coincide,
see Fig.~\ref{fig:snake}.
However, since we placed only one qubit at $w^{2r}$ and $w^1$,
the total number of qubits per site is at most two.

For each $r=1,\ldots,t$ define $2r$ additional gauge generators
$g_r^i$ and $h_r^i$ with $i=1,\ldots,r$ as shown below, see also Fig.~\ref{fig:generators}.
\begin{equation}
\label{weight2}
g_r^i=(w^{2i}+w^{2i+1})[D_r], 
\end{equation}
where  $i=1,\ldots,r-1$,
\begin{equation}
\label{weight2'}
g_r^r=(w^1+w^{2r})[D_r],
\end{equation}
\begin{equation}
\label{weight6}
h_r^i=(w^{2i-1}+w^{2i})[D_r]+(u^{2i-1}+u^{2i})[B_r]+ (v^{2i-1} +v^{2i})[A_{r-1}], 
\end{equation} 
where  $i=1,\ldots,r-1$, and
\begin{equation}
\label{weight6'}
h_r^r=(w^{2r-1}+w^{2r})[D_r] +(u^{2r-1}+u^{2r}+u^{2r+1})[B_r] +v^{2r-1}[A_{r-1}].
\end{equation}

\begin{figure}[h]
\centerline{\includegraphics[height=5cm]{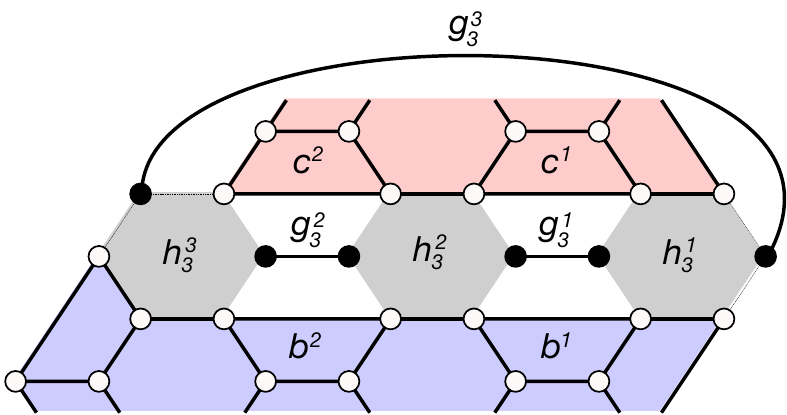}}
\caption{Additional gauge generators $h_r^i$ and $g_r^i$ for $r=3$.
The generator $h_r^i$ acts on $B$-qubits at the boundary of $\Lambda_r$ (blue)
and $A$-qubits at the boundary of $\Lambda_{r-1}$ (red). In addition, $h_r^i$ acts
on two $D$-qubits (filled circles). The generator $g_r^i$ acts only on two $D$-qubits.
The special faces $b^i$ and $c^i$ of the color code lattice
are used to define modified stabilizers.
The original stabilizer $b^i[A_r]+b^i[B_r]\in \calT_t$ must be replaced by
$g_r^i+b^i[A_r]+b^i[B_r]\in \calU_t$.
The original stabilizer $c^i[A_r]+c^i[B_r]\in \calT_t$ must be replaced by
$g_{r+1}^i+c^i[A_r]+c^i[B_r]\in \calU_t$.
\label{fig:generators}
}
\end{figure}

Note that all additional generators except for $g_r^r$ are spatially local. Furthermore, 
\begin{equation}
\label{bad_guy2}
\omega_{r,r-1}=\sum_{i=1}^r h_r^i +g_r^i.
\end{equation}
This identity will allow us to get rid of the non-local generators
$\omega_{r,r-1}$. More formally, 
define a subspace $\calU_t\subseteq \calE^{N_t}$ such that 
\begin{equation}
\label{Ut}
\dot{\calU}_t =\dot{\calT}_t +\sum_{r=1}^t \langle g_r^1,\ldots,g_r^r,h_r^1,\ldots,h_r^r\rangle.
\end{equation}
Note that $\dot{\calU}_t$ uniquely defines $\calU_t$ since applying the dot
operation twice gives the original subspace, see Section~\ref{sec:subs}.
We shall see that $\calU_t$ can be regarded as an extended version of $\calT_t$.
Similarly, define a subspace $\calD_t\subseteq \calE^{N_t}$ such that 
\begin{equation}
\label{Dt}
\dot{\calD}_t =\dot{\calC}_t +\sum_{r=1}^t \langle g_r^1,\ldots,g_r^r,h_r^1,\ldots,h_r^r\rangle.
\end{equation}
We shall see that $\calD_t$ can be regarded as an extended version of $\calC_t$.
Here it is understood that vectors from $\dot{\calT}_t$ and $\dot{\calC}_t$ are extended to $D$-qubits
by zeros.   From Eq.~(\ref{Cproperty3}) we infer that
\begin{equation}
\label{DU}
\dot{\calD}_t\subseteq \dot{\calU}_t\quad \mbox{and}  \quad  \calU_t\subseteq \calD_t.
\end{equation}
Let us establish some basic properties of $\calD_t$ and $\calU_t$.
\begin{lemma}
\label{lemma:gauge}
$\calU_t\subseteq \dot{\calU}_t$.
Furthermore, if $h\in \calU_t$ then the restriction of $h$ onto the union of $A$ and $B$ qubits
is contained in $\calT_t$.
\end{lemma}
\begin{lemma}
\label{lemma:gauge1}
$\calD_t\subseteq \dot{\calD_t}$.
Furthermore, if $h\in \calD_t$ then the restriction of $h$ onto the union of $A$ and $B$ qubits
is contained in $\calC_t$.
\end{lemma}
Since the proof of the two lemmas is identical, we only prove Lemma~\ref{lemma:gauge}.
\begin{proof}
Suppose $h\in \calU_t$. By definition of the dot operation,
$h$ has even weight and orthogonal to any vector of $\dot{\calU}_t$.
Let $h_D$ be a restriction of $h$ 
onto the union of all $D$-qubits.
Since $h_D$ is orthogonal to any generator $g_r^i$,
we infer that $h_D$ is a linear combination of $g_r^i$. 
This implies $h_D\in \dot{\calU}_t$. 
Let $h_{AB}$ be a restriction of $h$ onto the union of all $A$- and $B$-qubits.
Since $h\in \calE$ and $h_D\in \calE$,   we infer that $h_{AB}\in \calE$.
The inclusion  
$\dot{\calT}_t\subseteq \dot{\calU}_t$ implies 
$h_{AB}\in \dot{\calT}_t^\perp \cap \calE=\calT_t$.
This proves the second statement. Finally, $\calT_t\subseteq \dot{\calT}_t \subseteq \dot{\calU}_t$
implies $h_{AB}\in \dot{\calU}_t$, that is, $h\in \dot{\calU}_t$.
\end{proof}
Combining the lemmas and Eq.~(\ref{DU}) yields 
\begin{equation}
\label{DU1}
\calU_t\subseteq \calD_t \subseteq \dot{\calD}_t \subseteq \dot{\calU}_t.
\end{equation}
One can view Eq.~(\ref{DU1}) as an extended version of Eq.~(\ref{Cproperty3}).
We can now define extended versions of the doubled color codes, see Table~\ref{table:CTx}.

\begin{table}[!ht]
\centerline{
\begin{tabular}{r|c|c|c|}
 & Transversal gates  & Stabilizer group & Gauge group \\
\hline
$\vphantom{\hat{\hat{A}}}$
$C$-code & Clifford group &  $\css{\calD_t}{\calD_t}$ & $\css{\dot{\calD}_t}{\dot{\calD}_t} $   \\
\hline
$\vphantom{\hat{\hat{A}}}$
$T$-code & $T$ gate &  $\css{\calU_t}{\dot{\calU}_t}$ & $\css{\calU_t}{\dot{\calU}_t}$  \\
\hline
$\vphantom{\hat{\hat{A}}}$
Base code &   & $\css{\calU_t}{\calD_t}$  & $\css{\dot{\calD}_t}{\dot{\calU}_t}$ \\
\hline
\end{tabular}}
\caption{Extended doubled color codes.}
\label{table:CTx}
\end{table}
As before, we define the base code such that its stabilizer group
is the intersection of stabilizer groups of all codes in the family. 
Lemma~\ref{lemma:gauge1} implies that $\calD_t$ is self-orthogonal,
that is, the group $\css{\calD_t}{\calD_t}$ is abelian and the $C$-code
is well-defined. Likewise, Eq.~(\ref{DU1}) implies that 
$\calU_t$ and $\calD_t$ are mutually orthogonal, so that 
the group $\css{\calU_t}{\calD_t}$ is abelian and the base code is well-defined.
Since we already know that $\calC_t$ is doubly even and $\calT_t$ is triply even,
Lemmas~\ref{lemma:gauge},\ref{lemma:gauge1} have the following corollary.
\begin{corol}
\label{corol:DU}
The subspace $\calU_t$ is triply even with respect to the same subsets as $\calT_t$.
The subspace $\calD_t$ is doubly even with respect to the same subsets as $\calC_t$.
\end{corol}
This shows that the extended $C$ and $T$ codes have transversal
Clifford gates and the $T$-gate respectively. 
Let us  describe gauge generators of the extended $C$-code.
From  Eq.~(\ref{bad_guy2}) we infer that the non-local generators $\omega_{r,r-1}$
can be removed from the generating set of $\dot{\calD}_t$. Thus the extended $C$-code
has only face-type gauge generators and the additional generators 
$g_r^1,\ldots,g_r^r,h_r^1,\ldots,h_r^r$. The latter are spatially local except for $g_r^r$.
A direct inspection of Eq.~(\ref{Dt}) shows that 
generators of $\dot{\calD}_t$ can be partitioned into mutually disjoint subsets
supported on regions 
\[
M_t= A_t, \quad M_{t-1}=B_tD_tA_{t-1}, \quad M_{t-2}=B_{t-1} D_{t-1} A_{t-2},
\quad \ldots \quad
, \quad M_0=B_1 D_1A_0.
\]
Thus we can analyze properties of the extended $C$-code on each region $M_r$ separately.
Generators of $\dot{\calD}_t$  supported on $M_t$ have
a form $f[A_t]$ with $f\in \calS_t$. 
These generators describe the regular color code $\calS_t$.
Thus the extended $C$-code can be converted to the regular color
code $\css{\calS_t}{\calS_t}$ by discarding all the regions except for $M_t$.
As before, the extended $T$-code is obtained from the extended $C$-code
by adding edge-type $Z$-stabilizers.

Next let us prove that the extended codes
have distance $2t+1$. Combining  Eq.~(\ref{dAdB}) and Eq.~(\ref{DU1}) one
can see that the distance of any code defined in Table~\ref{table:CTx}
is lower bounded by $d(\calU_t)$. Thus it suffices to prove that 
$d(\calU_t)=2t+1$.
We shall need an explicit expression for generators of $\calU_t$. Note that
\begin{equation}
\label{gall}
g_r\equiv \sum_{i=1}^r g_r^i \in \calU_t.
\end{equation}
Indeed, $g_r$ does not overlap with generators of $\dot{\calT}_t $
and  has even overlap with all additional generators $g_r^i$ and $h_r^i$.
Next consider a generator $g_r^i$ with $1\le i<r$. 
Let $b^i\in \calS_r$ be the face of $\Lambda_r$ located directly below
$g_r^i$, see Fig.~\ref{fig:generators}. 
We claim that 
\begin{equation}
\label{faces1}
\beta_r^i\equiv g_r^i+b^i[A_r] + b^i[B_r] \in \calU_t 
\end{equation}
for all $r=1,\ldots,t$ and all $i=1,\ldots,r-1$.
Indeed, $\beta_r^i$ has even overlap with all generators of $\dot{\calT}_t$
since $b^i[A_r] + b^i[B_r]\in \calT_t$. Furthermore, $\beta_r^i$ has
even overlap with all additional generators $g_{r'}^{i'}$. 
 It remains to check that $\beta_r^i$ has even overlap
with the additional generators  $h_{r'}^{i'}$. The only non-trivial case is $h_r^j$ with $j=i$
or $j=i+1$, see Fig.~\ref{fig:generators}.
In this case both $g_r^i$ and $b^i[A_r]+b^i[B_r]$
have odd overlap with $h_r^j$, so that
$\beta_r^i$ has even overlap with $h_r^j$. This proves that
$\beta_r^i\in \calU_t$. 

Likewise, let $c^i\in \calS_{t-1}$ be the face of $\Lambda_{r-1}$ located directly above $g_r^i$,
see Fig.~\ref{fig:generators}. The same arguments as above show that 
\begin{equation}
\label{faces2}
\gamma_r^i\equiv g_r^i+c^i[A_{r-1}] + c^i[B_{r-1}] \in \calU_t
\end{equation}
for all $r=2,\ldots,t$ and for all $i=1,\ldots,r-1$.
Note that the sublattice $\Lambda_t$ has only special faces $b^i$
whereas $\Lambda_1$ has only special faces $c^i$. 
All other sublattices $\Lambda_r$ have both types of special faces.
Now we are ready to describe generators of $\calU_t$.
\begin{lemma}
\label{lemma:nasty}
Suppose 
 $2\le r\le t-1$.
Let $\calS^*_r\subseteq \calS_r$ be the subspace
spanned by all faces of the color code lattice $\Lambda_r$ except for
the special faces $b^i$ and $c^i$.
Let $\calS^*_t\subseteq \calS_t$ be the subspace spanned by 
all faces of $\Lambda_t$ except for  $b^i$. 
Let $\calS^*_1\subseteq S_1$ be the 
subspace spanned by 
all faces of $\Lambda_1$ except for  $c^1$.
Then
\begin{align}
\label{nasty}
\calU_t &=\sum_{r=1}^t \langle f[A_r]+f[B_r]\, : \, f\in \calS_r^*\rangle +
 \langle B_r A_{r-1} \rangle \nonumber \\
& + \sum_{r=1}^t \langle g_r\rangle + \langle \beta_r^i,\gamma_r^i \, : \, i=1,\ldots,r-1\rangle.
\end{align}
\end{lemma}
\begin{proof}
We have already shown that the last two terms in Eq.~(\ref{nasty}) are 
contained in $\calU_t$. A direct inspection shows that 
vectors $f[A_r]+f[B_r]$ with $f\in \calS_r^*$ and 
$B_r A_{r-1}$  have even overlap with all additional generators
$g_j^i$, $h_j^i$. Furthermore, $f[A_r]+f[B_r]$ and  $B_r A_{r-1}$  are contained in $\calT_t$ and 
thus have even overlap with any vector in $\dot{\calT}_t$. 
This proves the inclusion $\supseteq$ in Eq.~(\ref{nasty}).

Conversely, consider any $k\in \calU_t$. 
We have to prove that $k$ is contained in the  righthand
side of Eq.~(\ref{nasty}).
The same arguments as in the proof of Lemma~\ref{lemma:gauge} show that
\begin{equation}
\label{kk1}
k=k_{AB}+k_D, \quad k_{AB}\in \calT_t, \quad k_D=\sum_{r=1}^t \sum_{i=1}^r x_r^i g_r^i,
\end{equation}
where $k_{AB}$ and $k_D$  have support only on $AB$-qubits and $D$-qubits respectively.
Here $x_r^i\in \FF_2$ are some coefficients. 
Let us modify $k$ according to 
\[
k\gets k+\sum_{r=1}^t \sum_{i=1}^r x_r^i \beta_r^i.
\]
We still have the inclusion $k\in \calU_k$ since $\beta_r^i\in \calU_k$, see above. 
Furthermore, the term $x_r^i \beta_r^i$ cancels the term $x_r^i g_r^i$ in $k_D$
and modifies the term $k_{AB}$ according to $k_{AB}\gets k_{AB}+b^i[A_r]+b^i[B_r]$.
Now $k$ has support only on  $AB$-qubits
and $k\in \calT_t$.
Let us express $k$ as a sum of generators of $\calT_t$
defined in  Eq.~(\ref{Tgens}). This yields
\[
k=k^* + \sum_{r=2}^t \sum_{i=1}^{r-1} y_r^i (b^i[A_r] + b^i[B_r]) + z_r^i (c^i[A_{r-1}] + c^i[B_{r-1}])
\]
for some 
\[
k^*\in \sum_{r=1}^t \langle f[A_r]+f[B_r]\, : \, f\in \calS_r^*\rangle + 
\langle  B_r A_{r-1} \rangle
\]
and some coefficients $y_r^i,z_r^i\in \FF_2$.
Note that $k^*$ is contained in the righthand side of Eq.~(\ref{nasty}).
Furthermore, $k^*\in \calU_t$ which implies $k+k^*\in \calU_t$. 
In particular, 
$k+k^*$ has even overlap 
with all generators $h_r^j$.
On the other hand, $h_r^j$ has odd overlap with $b^i[A_r] + b^i[B_r]$
and with $c^i[A_{r-1}] + c^i[B_{r-1}]$ for $i=j,j-1$, see Fig.~\ref{fig:generators}.
Thus $k+k^*$ may have even overlap with $h_r^j$ only if $y_r^i=z_r^i$ for all $i$ and $r$.
Using the identity
\[
(b^i[A_r] + b^i[B_r]) +  (c^i[A_{r-1}] + c^i[B_{r-1}])=\beta_r^i + \gamma_r^i
\]
we conclude that $k+k^*$ is contained in the last term in Eq.~(\ref{nasty}).
Since $k^*$ is contained in the sum of the first two terms in Eq.~(\ref{nasty}),
we have proved the inclusion $\subseteq$ in Eq.~(\ref{nasty}).
 \end{proof}

The following lemma is the most difficult part of the proof. 
\begin{lemma}
\label{lemma:tricky}
$d(\calU_t)=d(\calT_t)=2t+1$.
\end{lemma}
\begin{proof}
We shall use induction in $t$. The base of induction is $t=1$. In this case
there are only two $D$-qubits, $D_1=\{w^1,w^2\}$,
and one  additional gauge generator $g_1\equiv g_1^1=w^1+w^2$.
Furthermore, $g_1$ is a stabilizer,  $g_1\in \calU_1$. 
Consider a minimum weight vector $k\in \calU_1^\perp \cap \calO$
such that $d(\calU_1)=|k|$.
Since $k$ has even overlap with $g_1$, one has
either $g_1\subseteq k$ or $g_1\cap k=\emptyset$.
The first case can be ruled out since $k+g_1\in \calU_1^\perp \cap \calO$
would have  weight $|k|-2$. Thus $k$ has support only on $AB$-qubits.
Lemma~\ref{lemma:gauge} implies that $\calU_1$ and $\calT_1$
have the same restriction on $AB$-qubits, that is, $d(\calU_1)=d(\calT_1)=3$.

Consider now some $t\ge 2$.
Let us rewrite the partition in Eq.~(\ref{ABD}) as 
\[
[N_t]=A_tB_tD_tC_t, \quad \quad C_t\equiv A_{t-1} B_{t-1} D_{t-1} \ldots  A_1 B_1 D_1 A_0.
\]
Consider an arbitrary vector $k\in \calU_t^\perp \cap \calO$. Let us write
\[
k=\alpha[A_t]+\beta[B_t]  + \delta[D_t] +  \gamma[C_t],
\]
for some vectors
\[
\alpha,\beta \in \FF_2^{m_t}, \quad \delta\in \FF_2^{2t}, \quad \gamma\in \FF_2^{N_{t-1}}.
\]
First we claim that 
\begin{equation}
\label{parity_stuff}
\alpha\in \calO, \quad \beta+\gamma\in \calE, \quad  \mbox{and} \quad
\delta \in \calE.
\end{equation}
Indeed, since $k$ has even overlap with the stabilizer $g_t$
with $\supp{g_t}=D_t$, 
see Eq.~(\ref{gall}), we infer that $0=\trn{k}g_t=|\delta|{\pmod 2}$,
that is, $\delta\in \calE$.
Furthermore, Lemma~\ref{lemma:nasty} implies 
\[
h\equiv B_tC_t
=\sum_{r=1}^t B_rA_{r-1} + \sum_{r=1}^{t-1} g_r\in \calU_t.
\]
Since $\supp{h}=B_tC_t$ and $0=\trn{k}h=|\beta|+|\gamma|{\pmod 2}$, 
we infer that $\beta+\gamma\in \calE$.
Finally, $\alpha\in \calO$ follows from the above and the assumption that $k\in \calO$. 

We shall refer to a substitution $k\gets k+g$ with $g\in \dot{\calU}_t$ as a
gauge transformation. Note that gauge transformations preserve the set $\calU_t^\perp \cap \calO$.
Our strategy will be to choose a sequence of gauge transformations
that transform $k$ into a form 
\begin{equation}
\label{desired}
k = e^i[A_t]+e^i[B_t] + \gamma[C_t]
\end{equation}
without increasing the weight of $k$. 
Here $e^i$ is some basis vector of $\FF_2^{m_t}$. 
Let us first assume that $k$ has the desired form Eq.~(\ref{desired})
and show that this implies $|k|\ge 2t+1$.
Indeed, using  Eq.~(\ref{nasty}) one can check that 
$e^i[A_t]+e^i[B_t]$ is orthogonal to all generators of $\calU_t$
except for $B_tA_{t-1}$.
Therefore if $k\in \calU_t^\perp$ has a special form  Eq.~(\ref{desired}) 
then $\gamma\in \calU_{t-1}^\perp$. Furthermore, $k\in \calO$
implies $\gamma\in \calO$, that is, $\gamma\in \calU_{t-1}^\perp \cap \calO$.
By induction hypothesis, $|\gamma|\ge 2t-1$
and thus $|k|\ge 2t+1$. 

It remains to show that $k$ can be transformed into the desired form Eq.~(\ref{desired})
without increasing the weight.
First, choose any $i\in \mathrm{supp}(\alpha)$ and let
$\alpha'=\alpha+e^i$. Note that $\alpha'\in \calE$ due to Eq.~(\ref{parity_stuff}),
so that  $\alpha'[A_t]+\alpha'[B_t]\in \dot{\calU}_t$.
In addition,  $|\alpha'|=|\alpha|-1$, so that 
\[
|k+\alpha'[A_t]+\alpha'[B_t]| = 1 + |\beta+\alpha'| +|\gamma|+|\delta|
\le 1+ |\alpha'|+|\beta|+|\gamma|+|\delta| =|\alpha|+|\beta|+|\gamma|+|\delta| =|k|.
\]
Thus we can transform $k$ to $k+\alpha'[A_t]+\alpha'[B_t]$ obtaining 
\[
k=e^i[A_t] + \beta[B_t] + \delta[D_t] + \gamma[C_t] 
\]
for some new vector $\beta\in \FF_2^{m_t}$.
Define $\theta=\beta+e^i$ so that 
\begin{equation}
\label{k1}
k=e^i[A_t] + e^i[B_t] + \theta[B_t] +  \delta[D_t] + \gamma[C_t].
\end{equation}
Consider first the case when $\theta=\overline{0}$.
Since $k$ has even overlap with the stabilizer $\beta_t^j$
and so does $e^i[A_t] + e^i[B_t]$, we conclude that 
$\delta[D_t]$ has even overlap with the gauge generator $g_t^j$
for any $j$. 
This is possible only if $\delta=\overline{0}$
or $\delta=\overline{1}$. If $\delta=\overline{0}$
then $k$ already has  the desired form Eq.~(\ref{desired}).
If $\delta=\overline{1}_W$ then $k+g_t^t$
has the desired form and $|k+g_t^t|\le |k|$.

Consider now the case  $\theta \ne \overline{0}$. 
Since $k$ has even overlap with the stabilizers $\beta_t^j$
and so does $e^i[A_t] + e^i[B_t]$, we conclude that 
$\delta[D_t]$ has even (odd) overlap with a generator  $g_t^j$
iff $\theta$ has even (odd) overlap with the special face $b^j$,
see Fig.~\ref{fig:generators}.
Let $\epsilon(\theta)\equiv |\theta| {\pmod 2}$ be the total parity of 
$\theta$. Let $\sigma(\theta)$ be the set of faces
of $\Lambda_t$ that have odd overlap with $\theta$.
The above shows that 
\begin{equation}
\label{gamma_syndrome}
\sigma(\theta)\subseteq \{ b^1,b^2,\ldots,b^{t-1}\}.
\end{equation}
Note that 
\begin{equation}
\label{gamma_syndrome1}
|\sigma(\theta)|=0{\pmod 2}
\end{equation}
since $\delta \in \calE$, see Eq.~(\ref{parity_stuff}).
We claim that one can choose $\theta^*\subseteq \Lambda_t$
such that 
\begin{equation}
\label{theta*}
|\theta^*|\le |\theta|, \quad \sigma(\theta^*)=\sigma(\theta), \quad \epsilon(\theta^*)=\epsilon(\theta),
\quad \mbox{and} \quad \mathrm{supp}(\theta^*)\subseteq \{u^1,u^2,\ldots,u^{2t+1}\}.
\end{equation}
Recall that $u^1,u^2,\ldots,u^{2t+1}$ are the sites of $\Lambda_t$ located
on the boundary facing $\Lambda_{t-1}$, see Fig.~\ref{fig:gadget}.
Indeed, it is straightforward to choose some $\theta^*$ that satisfies all above
conditions except for the first one. However, the property that $\theta^*$ is supported
on the boundary $\omega\equiv \{u^1,u^2,\ldots,u^{2t+1}\}$ implies that its weight cannot be 
decreased by adding faces of the lattice $\Lambda_t$. 
Indeed, the boundary $\omega$  is a minimum weight logical operator of the color code.
This implies $|\omega+f|\ge |\omega|$ for any stabilizer $f\in \calS_t$.
Equivalently, $|f\setminus \omega|\ge |f\cap \omega|$. 
This implies 
\[
|f+\theta^*|=|f\setminus \omega| + |(f\cap \omega)+\theta^*| 
\ge |f\cap \omega| +  |(f\cap \omega)+\theta^*|  \ge |\theta^*|.
\]
Finally, we note that  $\sigma(\theta^*)$ and $\epsilon(\theta^*)$
fix $\theta^*$ modulo stabilizers $f\in \calS_t$. This proves that we can satisfy all
conditions in Eq.~(\ref{theta*}).

Note that $\theta[B_t]+\theta^*[B_t]\in \dot{\calU}_t$ since 
$\theta+\theta^*\in \calS_t$.
A gauge transformation $k\to k+\theta[B_t]+\theta^*[B_t]$
can potentially {\em increase} the weight of $k$
if $i\in \theta$ but $i\notin \theta^*$. However, the weight can
increase at most by two. We will
``borrow" two units of weight keeping in mind 
that at least one of the subsequent gauge transformations
has to {\em decrease} the weight of $k$ at least by two,
so that  we  maintain a zero weight balance. 
Transforming 
$k$ to $k+\theta[B_t]+\theta^*[B_t]$ we obtain
\begin{equation}
\label{k2}
k=e^i[A_t] + e^i[B_t] + \theta^*[B_t] +  \delta[D_t] + \gamma[C_t].
\end{equation}
We have to two consider two cases.\\
\noindent
{\em Case~1:} $\theta^*\in \calE$.
Then  $u^1\notin \theta^*$ and $u^{2t+1}\notin \theta^*$, so that 
$\theta^*$ consists of disjoint paths
connecting consecutive pairs of faces in $\sigma(\theta^*)$. 
Consider any such path 
\[
\pi=u^{2i+1}+u^{2i+2}+\ldots+u^{2m-1}+u^{2m}\subseteq \theta^*
\]
that creates a pair of syndromes at faces $b^i$ and $b^m$ for some $i<m$.
An example of such path with $i=1$ and $m=4$ is shown on Fig.~\ref{fig:cleaning1}.
As we argued above, $\delta[D_t]$ must have odd overlap with $g_r^i$
and $g_r^m$.  Applying, if necessary, gauge transformations
$k\to k+g_r^i$ and $k\to k+g_r^m$ we can assume that 
$w^{2i+1}\in \delta$ and $w^{2m}\in \delta$. Then a gauge transformation
\[
k\to k+\sum_{p=i+1}^m h_t^p+ \sum_{p=i+1}^{m-1} g_t^p
\]
cleans $k$ out of all qubits  $\pi[B_t]$, qubits $(w^{2i+1}+w^{2m})[D_t]$, 
and potentially adds weight at $|\pi|$ qubits
\[
(v^{2i+1}+v^{2i+2}+\ldots+v^{2m-1}+v^{2m})[A_{t-1}] 
\]
Overall, the weight decreases at least by two, see Fig.~\ref{fig:cleaning1}.

\begin{figure}[h]
\centerline{\includegraphics[height=4.5cm]{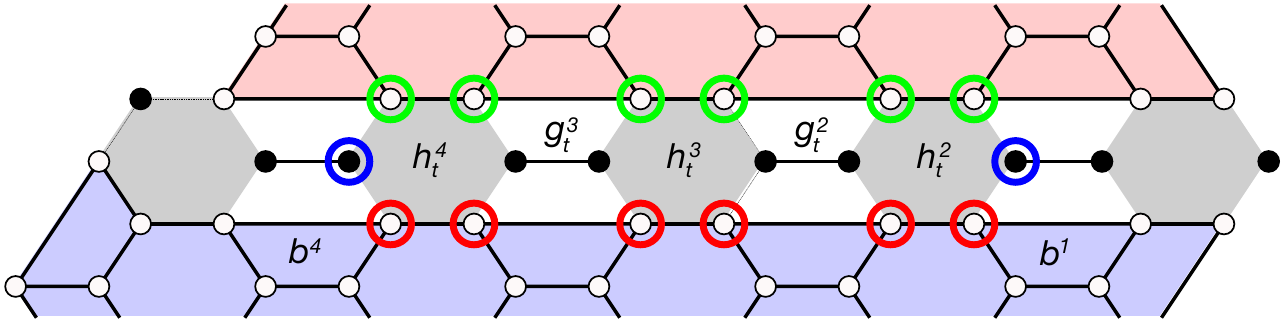}}
\caption{Example of a vector  $k$ defined  in Eq.~(\ref{k2})
such that $\theta^*\in \calE$ creates a pair of syndromes at faces
$b^1$ and $b^4$. Sites of $\theta^*$ are indicated by red circles.
In order to have zero syndrome for stabilizers $\beta_t^1$ and $\beta_t^4$,
the support of $k$ must also include sites $w^3$ and $w^8$ (possibly, after a 
gauge transformation $k\to k+g_t^1$ and $k\to k+g_t^4$). The sites 
$w^3$ and $w^8$  are indicated by blue circles. 
A gauge transformation $k\to k+h_t^2+h_t^3+h_t^4+g_t^2+g_t^3$
cleans out $k$ from all qubits indicated by blue and red circles,
potentially adding weight to qubits indicated by green circles. 
Overall, the weight of $k$ decreases at least by two.
Here $t=5$.
\label{fig:cleaning1}
}
\end{figure}

\noindent
{\em Case~2:} $\theta^*\in \calO$.
Then $u^1\in \theta^*$ and $u^{2t+1}\in \theta^*$,
so that $\theta^*$ consists of a path connecting the leftmost face
in $\sigma(\theta^*)$ to the site $u^{2t+1}$,
a path connecting the rightmost face in $\sigma(\theta^*)$ to the site $u^1$,
and, possibly, disjoint paths connecting consecutive pairs of faces in $\sigma(\theta^*)$,
see Fig.~\ref{fig:cleaning2}.
The latter can be cleaned out in the same fashion as in Case~(1), so below we focus on 
the former. Let $b^i$ be the rightmost face in $\sigma(\theta^*)$.
Then $\theta^*$ contains a path
\[
\pi_{right}=u^1+u^2+\ldots+u^{2i-1}+u^{2i}.
\]
Let $b^m$ be the leftmost face in $\gamma(\theta^*)$.
Then $\theta^*$ contains a path
\[
\pi_{left}=u^{2m+1}+u^{2m}+\ldots,+u^{2t}+u^{2t+1}.
\]
As we argued above, $\delta[D_t]$ must have odd overlap with $g_t^i$ and $g_t^m$.
Applying, if necessary, gauge transformations
$k\to k+g_t^i$ and $k\to k+g_t^m$ we can assume that 
$w^{2i}\in \delta$ and $w^{2m+1}\in \delta$. Then a gauge transformation
\[
k\to k+\sum_{p=1}^i h_t^p + \sum_{p=m+1}^{t-1} h_t^p + \sum_{p=1}^{i-1} g_t^p + \sum_{p=m+1}^t g_t^p
\]
cleans $k$ out of all qubits of $(\pi_{left}+\pi_{right})[B_t]$, 
qubits $(w^{2i}+w^{2m+1})[D_t]$, and potentially adds
weight at $|\pi_{right}|$ qubits $(v^1+\ldots+v^{2i})[A_{t-1}]$
and $|\pi_{left}|-2$ qubits $(v^{2m+1}+\ldots+v^{2t-1})[A_{t-1}]$.
Overall, the weight decreases at least by four, see Fig.~\ref{fig:cleaning2}.
In both cases, we transform $k$ to the desired form
and the weight decreases at least by two. 
\end{proof}

\begin{figure}[h]
\centerline{\includegraphics[height=6cm]{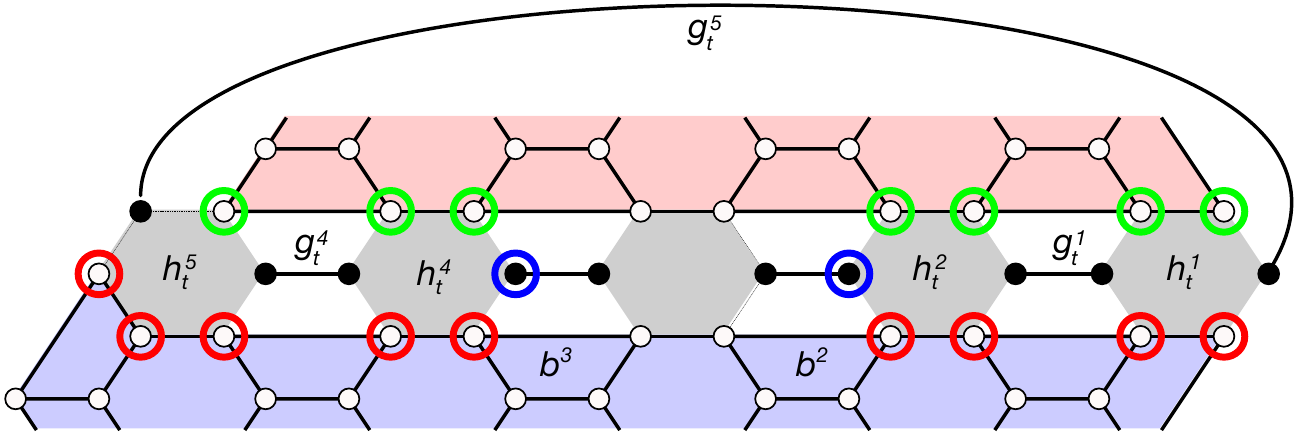}}
\caption{Example of a vector $k$ defined in Eq.~(\ref{k2})
such that $\theta^*\in \calO$ creates a pair of syndromes at faces
$b^2$ and $b^3$. Sites of $\theta^*$ are indicated by red circles.
In order to have zero syndrome for stabilizers $\beta_t^2$ and $\beta_t^3$,
the support of $k$ must also include sites $w^4$ and $w^7$ (possibly, after a 
gauge transformation $k\to k+g_t^2$ and $k\to k+g_t^3$). The sites 
$w^4$ and $w^7$  are indicated by blue circles. 
A gauge transformation $k\to k+h_t^1+h_t^2+h_t^4+h_t^5+g_t^1+g_t^4+ g_t^5$
cleans out $k$ from all qubits indicated by blue and red circles,
potentially adding weight to qubits indicated by green circles. 
Overall, the weight of $k$ decreases at least by four. Here $t=5$.
\label{fig:cleaning2}
}
\end{figure}

The final weight reduction step is to transform the long-range generators $g_r^r$  into a local form. 
We shall use a coding theory analogue of the
subdivision gadget used to simulate long-range spin interactions by short-range ones~\cite{Oliveira2008}.
Consider an arbitrary subspace $\calU \subseteq \calE^n$ such that $\calU\subseteq \dot{\calU}$
and $n$ is odd. It defines a CSS code with a gauge group $\css{\dot{\calU}}{\dot{\calU}}$ and a stabilizer group
$\css{\calU}{\calU}$. Note that $\calU\subseteq \dot{\calU}$ implies that 
$\calU$ is self-orthogonal, so that the stabilizer group is abelian and the code is well-defined. 
Suppose  $\dot{\calU}$ contains a weight-two vector
supported  on the first two qubits, $e^1+e^2\in \dot{\calU}$. 
We envision a scenario when qubits $1$ and $2$ occupy two remote
lattice locations (for example sites $w^1$ and $w^{2t}$ in the above construction),
such that $e^1+e^2$ is not spatially local.
Let us add two ancillary qubits labeled $a$ and $b$.  We 
envision that  all vectors $e^1+e^a$, $e^a+e^b$, and $e^b+e^2$ are
spatially local (or, at least, more local compared with $e^1+e^2$). 
Define a subspace $\calV\subseteq  \calE_2^{n+2}$ such that 
\begin{equation}
\label{gadget1}
\dot{\calV}=\langle e^1+e^a, e^a+e^b, e^b+e^2\rangle + \dot{\calU}.
\end{equation}
Here it is understood that vectors of $\dot{\calU}$ are extended to the ancillary qubits by zeroes. 
Using  analogues
of Lemmas~\ref{lemma:gauge},\ref{lemma:gauge1} one can easily show that
$\calV$ is self-orthogonal, so that $\calV$ 
 defines an abelian  stabilizer group
$\css{\calV}{\calV}$. Below we prove that the codes $\css{\calU}{\calU}$
and $\css{\calV}{\calV}$ have the 
same distance and the same transversality properties. 
\begin{lemma}[\bf Subdivision gadget]
\label{lemma:subdivision}
Consider the subspaces $\calU$ and $\calV$ as above.
Suppose $d(\calU)\ge 3$. Then $d(\calV)=d(\calU)$.
Furthermore, if $\calU$ is triply (doubly) even with respect to some subsets
$M^\pm \subseteq [n]$ then $\calV$ is triply (doubly) even with respect to the same
subsets. 
\end{lemma}
\begin{proof}
Let $\dot{\calU}_{12}$  and $\calU_{12}$ 
be the subspaces including all vectors $g\in \dot{\calU}$ and $g\in \calU$
respectively such that $g_1=g_2=0$, that is, $g$ does not include qubits
$1$ and $2$. By assumption, 
$e^1+e^2\in \dot{\calU}$ and thus 
$f_1=f_2$ for any  $f\in \calU$ since $f$ must have
even overlap with any element of $\dot{\calU}$.
There must be at least one vector $f\in \calU$
such that $f_1=f_2=1$ since otherwise  $d(\calU)=1$.
Thus $\calU$ can be represented as 
\begin{equation}
\label{UUU}
\calU=\langle e^1+e^2+g\rangle + \calU_{12}\quad \mbox{for some $g\in \dot{\calU}_{12}$}.
\end{equation}
We claim that 
\begin{equation}
\label{Nu}
\calV=\langle h\rangle + \calU_{12},
\quad \mbox{where} \quad h\equiv e^1+e^a+e^b+e^2+g.
\end{equation}
Let us first prove the inclusion $\supseteq$ in Eq.~(\ref{Nu}).
We have $h\in \dot{\calU}^\perp$ since $e^1+e^2+g\in \calU \subseteq \dot{\calU}^\perp$
and $e^a+e^b\in \dot{\calU}^\perp$ since no vector in $\dot{\calU}$ includes $a$ or $b$.
Taking into account  Eq.~(\ref{gadget1}) one gets $h\in \dot{\calV}^\perp$
and thus $h\in \calV$. The inclusion $\calU_{12}\subseteq \calV$ follows
trivially from the definitions. 
Next let us prove the inclusion $\subseteq$ in Eq.~(\ref{Nu}).
Consider any vector $k\in \calV$. We note that $k_1=k_a=k_b=k_2$
since $k$ must have even overlap with any vector in $\dot{\calV}$. 
Replacing, if necessary, $k$ by $k+h$ we can assume that $k_1=k_a=k_b=k_2=0$.
Since $\dot{\calU}\subseteq \dot{\calV}$ we infer that $k\in \dot{\calU}^\perp$.
Taking into account that $k$ has even weight
one gets $k\in \ddot{\calU}=\calU$, that is, $k\in \calU_{12}$.
We have proved Eq.~(\ref{Nu}).

Suppose  $f\in \calV^\perp \cap \calO$ is a minimum weight vector
such that $|f|=d(\calV)$. 
Then  $f$ includes at most one of the qubits $1,a,b,2$
since otherwise we would be able to reduce the weight of $f$ by a gauge transformation
$f\gets f+f'$, where $f'$ is contained in the first term in Eq.~(\ref{gadget1}).
Without loss of generality $f_a=f_b=f_2=0$.
From Eq.~(\ref{Nu}) we infer that 
$\trn{h}f=0$ and $f\in \calU_{12}^\perp$.
Since $f_a=f_b=0$, this implies that $f$ is orthogonal to
 $e^1+e^2+g$ and thus $f\in \calU^\perp$, see Eq.~(\ref{UUU}).
 This shows that $f\in \calU^\perp\cap \calO$, that is, $d(\calV)=|f|\ge d(\calU)$.
The opposite inequality, $d(\calV)\le d(\calU)$, is obvious since 
extending any vector $f\in \calU^\perp \cap \calO$ by zeroes to qubits $a$ and $b$
gives a vector $f\in \calV^\perp \cap \calO$.
The last statement of the lemma follows 
from the fact that $\calU$ and $\calV$ have the same restriction 
on any subset $M^\pm \subseteq [n]$,
see Eq.~(\ref{Nu}).
\end{proof}

It remains to 
apply the subdivision gadget to the long-range gauge generators
$g_r^r=(w^1+w^{2r})[D_r]$, where $r=2,\ldots,t$
(note that $g_1^1$ is already spatially local).
Let us add a second qubit at each site $w^i$
of the region $D_r$, except for $w^1$ and $w^{2r}$. 
We shall denote these extra qubits as $\bar{w}^2,\ldots,\bar{w}^{2r-1}$
such that 
$D_r=\{w^1,w^2,\ldots,w^{2r},\bar{w}^2,\ldots,\bar{w}^{2r-1} \}$.
Qubits $w^i,\bar{w}^i$ share the same site of the lattice.
The  total number of qubits  becomes
\[
K_t=N_t+\sum_{r=2}^t (2r-2)=2t^3+8t^2+6t+1.
\]
Define a subspace $\calV_t\subseteq \calE_2^{K_t}$ such that 
\begin{equation}
\label{dotV}
\dot{\calV}_t=\dot{\calU}_t + \sum_{r=2}^t \langle (w^1+ \bar{w}^2)[D_r]\rangle 
+ \langle (\bar{w}^{2r-1} + w^{2r})[D_r]\rangle
+ \sum_{i=2}^{2r-2} 
\langle (\bar{w}^i+\bar{w}^{i+1})[D_r]\rangle.
\end{equation}
Here it is understood that vectors of $\dot{\calU}_t$ are extended to the added
qubits by zeroes.  Note that all generators of $\dot{\calV}_t$ are spatially local
since  the long-range generator $g_r^r$ can be decomposed as 
\begin{equation}
\label{bad_guy5}
g_r^r=w^1+w^{2r}= (w_1+\bar{w}_2) + (\bar{w}_2+\bar{w}_3) + \ldots + (\bar{w}_{2r-1}+w_{2r}),
\end{equation}
where each term has support on a single face of the lattice. 
Here we omitted $[D_r]$ to simplify notations.
We shall regard $\calV_t$ as an extended version of $\calU_t$. 
Similarly, define a subspace $\calF_t\subseteq \calE_2^{K_t}$ such that 
\begin{equation}
\label{dotF}
\dot{\calF}_t=\dot{\calD}_t + \sum_{r=2}^t \langle (w^1+ \bar{w}^2)[D_r]\rangle 
+ \langle (\bar{w}^{2r-1} + w^{2r})[D_r]\rangle
+ \sum_{i=2}^{2r-2} 
\langle (\bar{w}^i+\bar{w}^{i+1})[D_r]\rangle.
\end{equation}
We shall regard $\calF_t$ as an extended version of $\calD_t$. 
Note that all generators of $\dot{\calF}_t$ are spatially local. 
Using Eq.~(\ref{DU1}) and analogues
of Lemmas~\ref{lemma:gauge},\ref{lemma:gauge1} one can easily show that
\begin{equation}
\label{NuF}
\calV_t\subseteq \calF_t \subseteq \dot{\calF}_t \subseteq \dot{\calV}_t.
\end{equation}
We are now ready to define the final version of the doubled color codes,
see Table~\ref{table:CTxx}, 
that satisfy all the properties announced in Section~\ref{sec:summary}
(up to relabeling of the subspaces $\calC_t,\calT_t$ into 
$\calF_t,\calV_t$).
 Combining  Eq.~(\ref{dAdB}) and Eq.~(\ref{NuF}) one
can see that the distance of any code defined in Table~\ref{table:CTxx}
is lower bounded by $d(\calV_t)$.
A recursive application of the subdivision gadget  lemma shows that
$d(\calV_t)=d(\calU_t)=2t+1$.
Furthermore,  $\calV_t$ is triply-even
with respect to the same subsets as $\calU_t$
and $\calF_t$ is doubly even with respect to the same subsets as $\calD_t$.
Thus the final $C$ and $T$ codes have transversal Clifford gates
and the $T$-gate respectively.
\begin{table}[!ht]
\centerline{
\begin{tabular}{r|c|c|c|}
 & Transversal gates  & Stabilizer group & Gauge group \\
\hline
$\vphantom{\hat{\hat{A}}}$
$C$-code & Clifford group &  $\css{\calF_t}{\calF_t}$ & $\css{\dot{\calF}_t}{\dot{\calF}_t} $   \\
\hline
$\vphantom{\hat{\hat{A}}}$
$T$-code & $T$ gate &  $\css{\calV_t}{\dot{\calV}_t}$ & $\css{\calV_t}{\dot{\calV}_t}$  \\
\hline
$\vphantom{\hat{\hat{A}}}$
Base code &   & $\css{\calV_t}{\calF_t}$  & $\css{\dot{\calF}_t}{\dot{\calV}_t}$ \\
\hline
\end{tabular}}
\caption{Final version of the doubled color codes. All gauge generators
of the $C$-code are spatially local. The $T$-code can be obtained from
the $C$-code by measuring syndromes of edge-type stabilizers.
The base code is defined such that its stabilizer group
is the intersection of all other stabilizer groups.
}
\label{table:CTxx}
\end{table}
Let us  describe gauge generators of the final $C$-code.
From  Eq.~(\ref{bad_guy5}) we infer that the long-range generators $g_r^r$
can be removed from the generating set. Thus the final $C$-code
has only face-type gauge generators, the additional generators 
$g_r^1,\ldots,g_r^{r-1},h_r^1,\ldots,h_r^r$, and the additional generators
that appear in Eq.~(\ref{dotF}). All these generators are spatially local. 
The same arguments as above show that a restriction of the final $C$-code 
onto the region $A_t$ coincides with the regular color code $\css{\calS_t}{\calS_t}$.
As before, the final $T$-code is obtained from the extended $C$-code
by adding edge-type stabilizer generators.

Finally, let us point out that  the gauge generator $\omega_{1,0}\in \dot{\calT}_t$ is already spatially local,
so that we do not have to apply the weight reduction steps to $\omega_{1,0}$.
Thus the total number of physical qubits can be reduced from $K_t$ to  $K_t-2=2t^3+8t^2+6t-1$.

\section*{Acknowledgments}
SB thanks Andrew Landahl for fruitful discussions and helpful suggestions
at the early stages of this project.
SB acknowledges NSF grant PHY-1415461.


\end{document}